\documentclass[a4paper,11pt]{article}
\def\withcolors{0}
\def\withnotes{0}
\usepackage[T1]{fontenc}
\usepackage[utf8]{inputenc}

\usepackage{xspace}                                     

\usepackage[normalem]{ulem}

\usepackage{amsfonts,amsmath,amssymb, amsthm, mathtools}
\usepackage{thm-restate}

\usepackage{algorithmicx,algpseudocode,algorithm}

\usepackage[usenames,dvipsnames,table]{xcolor}

\usepackage{csquotes}

\usepackage{relsize}

\usepackage{multirow}
\usepackage{chngpage} 

\usepackage{mfirstuc}

\usepackage{tikz}
\usetikzlibrary{arrows}
\usetikzlibrary{calc,decorations.pathmorphing,patterns}

\usepackage{aliascnt}

\usepackage{fullpage}

\usepackage{titling}

\usepackage[shortlabels]{enumitem}
\setitemize{noitemsep,topsep=3pt,parsep=2pt,partopsep=2pt} 
\setenumerate{itemsep=1pt,topsep=2pt,parsep=2pt,partopsep=2pt}
\setdescription{itemsep=1pt}

\ifnum\withnotes=1
\usepackage[colorinlistoftodos,textsize=scriptsize]{todonotes}
\fi

\usepackage{verbatim}
\usepackage{bigfoot} 

\usepackage{pgfplots}

\makeatletter
\@ifundefined{theorem}{%
	\theoremstyle{plain} 
	\newtheorem{theorem}{Theorem}[section]
	\newaliascnt{coro}{theorem}
	\newtheorem{corollary}[coro]{Corollary}
	\aliascntresetthe{coro}
	\newaliascnt{lem}{theorem}
	\newtheorem{lemma}[lem]{Lemma}
	\aliascntresetthe{lem}
	\newaliascnt{clm}{theorem}
	
	\aliascntresetthe{clm}
	\newaliascnt{fact}{theorem}
	\newtheorem{fact}[theorem]{Fact}
	\aliascntresetthe{fact}
	
	\newaliascnt{prop}{theorem}
	
	\aliascntresetthe{prop}
	\newaliascnt{obs}{theorem}
	\newtheorem{observation}[obs]{Observation}
	\aliascntresetthe{obs}
	\newaliascnt{conj}{theorem}
	
	\aliascntresetthe{conj}
	
	\theoremstyle{remark} 
	\newtheorem{remark}[theorem]{Remark}

	\theoremstyle{definition} 
	\newaliascnt{defn}{theorem}
	\newtheorem{definition}[defn]{Definition}
	\aliascntresetthe{defn}
}{}
\makeatother

\providecommand{\email}[1]{\href{mailto:#1}{\nolinkurl{#1}\xspace}}

\ifnum\withcolors=1
\else

\fi

\ifnum\withnotes=1

\else

\fi
\newcommand{\ignore}[1]{\leavevmode\unskip} 

\newcommand\restr[2]{{
		\left.\kern-\nulldelimiterspace 
		#1 
		\vphantom{\big|} 
		\right|_{#2} 
}}

\makeatletter
\newcommand\ackname{Acknowledgements}
\if@titlepage
\newenvironment{acknowledgements}{%
	\titlepage
	\null\vfil
	\@beginparpenalty\@lowpenalty
	\begin{center}%
		\bfseries \ackname
		\@endparpenalty\@M
\end{center}}%
{\par\vfil\null\endtitlepage}
\else

\fi
\makeatother

\begin{document}

\title{Earthmover Resilience and Testing in Ordered Structures}
\author{Omri Ben-Eliezer\thanks{Blavatnik School of
		Computer Science, Tel Aviv University, Tel Aviv 69978, Israel.
		{\tt omrib@mail.tau.ac.il}.} ,
	Eldar Fischer\thanks{Faculty of Computer Science, Israel Institute of Technology (Technion), Haifa, Israel. {\tt eldar@cs.technion.ac.il}.}}

\begin{titlepage}
\date{ }
\maketitle
\thispagestyle{empty}
\begin{abstract}
One of the main challenges in property testing is to characterize those properties that are testable with a constant number of queries. For unordered structures such as graphs and hypergraphs this task has been mostly settled. However, for ordered structures such as strings, images, and ordered graphs, the characterization problem seems very difficult in general.

In this paper, we identify a wide class of properties of ordered structures -- the \emph{earthmover resilient} (ER) properties -- and show that the ``good behavior'' of such properties allows us to obtain general testability results that are similar to (and more general than) those of unordered graphs. A property $\mathcal{P}$ is ER if, roughly speaking, slight changes in the order of the elements in an object satisfying $\mathcal{P}$ cannot make this object far from $\mathcal{P}$.
The class of ER properties includes, e.g., all unordered graph properties, many natural visual properties of images, such as convexity, and all hereditary properties of ordered graphs and images. 

A special case of our results implies, building on a recent result of Alon and the authors, that the distance of a given image or ordered graph from \emph{any} hereditary property can be estimated (with good probability) up to a constant additive error, using a constant number of queries. 

\end{abstract}
\end{titlepage}

\section{Introduction} 
\label{sec:introduction}
\emph{Property testing} is mainly concerned with understanding the amount of information one needs to extract from an unknown input function $f$ to approximately determine whether the function satisfies a property $\mathcal{P}$ or is far from satisfying it. In this paper, the types of functions we consider are \emph{strings} $f \colon [n] \to \Sigma$; \emph{images} or \emph{matrices} $f \colon [m] \times [n] \to \Sigma$; and \emph{edge-colored graphs} $f \colon \binom{[n]}{2} \to \Sigma$, where the set of possible colors for each edge is $\Sigma$. In all cases $\Sigma$ is a \emph{finite} alphabet. Note that the usual notion of a graph corresponds to the special case where $|\Sigma| = 2$.

The systematic study of property testing was initiated by Rubinfeld and Sudan \cite{RubinfeldSudan1996}, and Goldreich, Goldwasser and Ron \cite{GoldreichGoldwasserRon1998} were the first to study property testing of combinatorial structures. An $\epsilon$-\emph{test} for a property $\mathcal{P}$ of functions $f:X \to \Sigma$ is an algorithm that, given query access to an unknown input function $f$, distinguishes with good probability (say, with probability $2/3$) between the case that $f$ satisfies $\mathcal{P}$ and the case that $f$ is \emph{$\epsilon$-far} from $\mathcal{P}$; the latter meaning that one needs to change the values of at least an $\epsilon$-fraction of the entries of $f$ to make it satisfy $\mathcal{P}$. 
In an $n$-vertex graph, for example, changing an $\epsilon$-fraction of the representation means adding or removing $\epsilon \binom{n}{2}$ edges. (The representation model we consider here for graphs is the adjacency matrix. This is known as the \emph{dense model}.)

In many cases, such as that of visual properties of images (where the input is often noisy to some extent), it is more natural to consider a robust variant of tests, that is \emph{tolerant} to noise in the input. Such tests were first considered by Parnas, Ron and Rubinfeld \cite{ParnasRonRubinfeld2006}. 
A test is \emph{$(\epsilon, \delta)$-tolerant} for some $0 \leq \delta(\epsilon) < \epsilon$ if it distinguishes, with good probability, between inputs that are \emph{$\epsilon$-far} from satisfying $\mathcal{P}$ and those that are \emph{$\delta(\epsilon)$-close} to (i.e., not $\delta(\epsilon)$-far from) satisfying $\mathcal{P}$.

One of the main goals in property testing is to characterize properties in terms of the number of queries required by an optimal test for them. 
If a property $\mathcal{P}$ has, for any $\epsilon > 0$, an $\epsilon$-test that makes a constant number of queries, depending only on $\epsilon$ and not on the size of the input, then $\mathcal{P}$ is said to be \emph{testable}. $\mathcal{P}$ is \emph{tolerantly testable} if for any $\epsilon > 0$ it has a constant-query $(\epsilon, \delta)$-test for some $0 < \delta(\epsilon) < \epsilon$. Finally, $\mathcal{P}$ is \emph{estimable} if it has a constant query $(\epsilon, \delta)$-test for \emph{any} choice of $0 < \delta(\epsilon) < \epsilon$. In other words, $\mathcal{P}$ is estimable if the distance of an input to satisfying $\mathcal{P}$ can be estimated up to a constant error, with good probability, using a constant number of queries.   

The meta-question that we consider in this paper is the following.
\begin{center}
\emph{What makes a certain property $\mathcal{P}$ testable, tolerantly testable, or estimable?}
\end{center}

\subsection{Previous works: Characterizations of graphs and hypergraphs}
\label{subsec:characterization_unordered}

For graphs, it was shown by Fischer and Newman \cite{FischerNewman2005} that the above three notions are equivalent, i.e., any testable graph property is estimable (and thus trivially also tolerantly testable).
A combinatorial characterization of the testable graph properties was 
obtained by Alon, Fischer, Newman and Shapira \cite{AlonFischerNewmanShapira2009} and analytic characterizations were obtained independently by Borgs, Chayes, Lov\'asz, S\'os, Szegedy and Vesztergombi  \cite{BorgsCLSSV2006} and Lov\'asz and Szegedy \cite{LovaszSzegedy2010} through the study of graph limits. 
The combinatorial characterization relates testability with \emph{regular reducibility}, meaning, roughly speaking, that a graph property $\mathcal{P}$ is testable (or estimable) if and only if satisfying $\mathcal{P}$ is equivalent to approximately having one of finitely many prescribed types of Szemer\'edi regular partitions \cite{Szemeredi1976}. A formal definition of regular reducibility is given in Section \ref{sec:prelims}.

Very recently, a similar characterization for hypergraphs was obtained by Joos, Kim, K\"uhn and Osthus \cite{JoosKimKuhnOsthus2017}, who proved that as in the graph case, 
testability, estimability and regular reducibility are equivalent for any hypergraph property.

A (partial) characterization of the graph properties $\mathcal{P}$ that have a constant-query test whose error is one-sided (i.e., tests that always accept inputs satisfying $\mathcal{P}$) was obtained by Alon and Shapira \cite{AlonShapira2008}. They showed that the only properties testable using an important and natural type of one-sided tests, that are \emph{oblivious} to the input size, are essentially the \emph{hereditary} properties. 

The above characterizations for graphs rely on a conversion of tests into \emph{canonical tests}, due to Goldreich and Trevisan \cite{GoldreichTrevisan2003}. A canonical test $T$ always behaves as follows: First it picks a set $U$ of vertices non-adaptively and uniformly at random in the input graph $G$, and queries all pairs of these vertices, to get the induced subgraph $G[U]$. Then $T$ decides whether to accept or reject the input deterministically, based only on the identity of $G[U]$ and the size of $G$. The number of queries needed by the canonical test is only polynomial in the number of queries required by the original test, implying that any testable property is also canonically testable.

To summarize, all of the following conditions are equivalent for graphs: Testability, tolerant testability, canonical testability, estimability, and regular reducibility.

\subsection{From unordered to ordered structures}
\label{subsec:unordered_to_ordered}
Common to all of the above characterization results is the fact that they apply to unlabeled graphs and hypergraphs, which are \emph{unordered} structures: Graph (and hypergraph) properties are symmetric in the sense that they are invariant under any relabeling (or equivalently, reordering) of the vertices. That is, if a labeled graph $G$ satisfies an unordered graph property $\mathcal{P}$, then any graph resulting from $G$ by changing the labels of the vertices is isomorphic to $G$ (as an unordered graph), and so it satisfies $\mathcal{P}$ as well.

A natural question that one may ask is whether similar characterizations hold for the more general setting of \emph{ordered} structures over a finite alphabet, such as \emph{images} and \emph{vertex-ordered graphs} in the two-dimensional case, and \emph{strings} in the one-dimensional case. While an unordered property is defined as a family of (satisfying) instances that is closed under relabeling, in the ordered setting, \emph{any} family of instances is considered a valid property.
The ordered setting is indeed much more general than the unordered one, as best exemplified by string properties: On one hand, unordered string properties are essentially properties of distributions over the alphabet $\Sigma$. On the other hand, \emph{any} property of \emph{any} finite discrete structure can be encoded as an ordered string property!

In general, the answer to the above question is \emph{negative}. It is easy to construct simple string properties that are testable and even estimable, but are neither canonically testable nor regular reducible.\footnote{To this end, canonical tests in ordered structures are similar to their unordered counterparts, but they act in an order-preserving manner. For example, a $q$-query test for a property $\mathcal{P}$ of strings $f \colon [n] \to \Sigma$ is \emph{canonical} if, given an unknown string $f \colon [n] \to \Sigma$, the test picks $q$ entries $x_1 < \ldots < x_q \in [n]$, queries them to get the values $y_1 = f(x_1), \ldots, y_q = f(x_q)$, and decides whether to accept or reject the input only based on the tuple $(y_1, \ldots, y_q)$. Canonical tests in ordered graphs or images are defined similarly, but instead of querying a random substring, we query a random induced ordered subgraph or a random submatrix, respectively.}
As an example, consider the binary string property $\mathcal{P}_{111}$ of ``not containing three consecutive ones''. The following is an $\epsilon$-test for $\mathcal{P}_{111}$ (estimation is done similarly): Pick a random consecutive substring $S$ of the input, of length $O(1/\epsilon)$, and accept if and only if $S$ satisfies $\mathcal{P}_{111}$. On the other hand, global notions like canonical testability and regular reducibility cannot capture the local nature of $\mathcal{P}$.
Moreover, it was shown by Fischer and Fortnow \cite{FischerFortnow2006}, building on ideas from probabilistically checkable proofs of proximity (PCPP), that there exist testable properties that are not tolerantly testable, as opposed to the situation in unordered graphs \cite{FischerNewman2005}. 

However, it may still be possible that a \emph{positive} answer holds for the above question if we restrict our view to a class of ``well behaved'' properties. 

\begin{center}
	\emph{Does there exist a class of properties that is {\bf wide} enough to capture many interesting properties, yet {\bf well behaved} enough to allow simple characterizations for testability?}
\end{center}

So far, we have seen that in general, properties in which the exact location of entries is important to some extent, like $\mathcal{P}_{111}$ and the property from \cite{FischerFortnow2006}, do not admit characterizations of testability that are similar to those of unordered graphs. But what about properties that are ultimately global? Can one find, say, an ordered graph property that is canonically testable but not estimable, for example? Stated differently,
\begin{center}
	\emph{Do the characterizations of testability in unordered graphs have analogues for canonical testability in ordered graphs and images?}
\end{center}

\subsection{Our contributions}
In this paper, we provide a partial positive answer to the first question, and a more complete positive answer to the second question. For the second question, we show that canonical testability in ordered graphs and images implies estimability and is equivalent to (an ordered version of) regular reducibility, similarly to the case in unordered graphs.
Addressing the first question,
 we identify a wide class of well-behaved properties of ordered structures, called the \emph{earthmover resilient} (ER) properties, providing characterizations of tolerant testability and estimability for these properties.
\paragraph{Earthmover resilient properties}
Roughly speaking, a property $\mathcal{P}$ of a certain type of functions is \emph{earthmover resilient} if slight changes in the \emph{order} of the ``base elements''\footnote{The base elements in an ordered graph are the vertices, and in images these are the rows and the columns; in strings the base elements are the entries themselves.} of a function $f$ satisfying $\mathcal{P}$ cannot turn $f$ into a function that is \emph{far} from satisfying $\mathcal{P}$. 
The class of ER properties captures several types of interesting properties: 
\begin{enumerate}
	\item Trivially, \emph{all} properties of \emph{unordered} graphs and hypergraphs.
	
	\item Global \emph{visual} properties of images. In particular, this includes any property $\mathcal{P}$ of black-white images satisfying the following: Any image $I$ satisfying $\mathcal{P}$ has a sparse black-white boundary. This includes, as special cases, properties like convexity and being a half plane, which were previously investigated in \cite{BermanMR2016, BermanMurzabulatovR2016, CanonneGGKW2016, CFSS2017, Raskhodnikova2003}. 
	See Subsection \ref{subsec:earthmover_resilience} for the precise definitions and statement and Appendix \ref{app:sparse_boundary} for the proof.
	
	\item All \emph{hereditary} properties of ordered graphs and images, as implied by a recent result of Alon and the authors \cite{AlonBenEliezerFischer2017}. While all hereditary unordered graph properties obviously fit under this category, it also includes interesting order-based properties, such as the widely investigated property of \emph{monotonicity} (see \cite{DodisGLRRS1999, ErgunKannanKRV2000} for results on strings and images over a finite alphabet), $k$-monotonicity \cite{CanonneGGKW2016}, forbidden poset type problems \cite{FischerNewman2007}, and more generally forbidden submatrix type problems \cite{AlonBenEliezer2017, AlonBenEliezerFischer2017, AlonFischerNewman2007, FischerRozenberg2007}.

\end{enumerate}
\paragraph{The new results}
ER properties behave well enough to allow us to fully characterize the \emph{tolerantly testable} properties among them in images and ordered graphs. In strings, it turns out that earthmover resilience is \emph{equivalent} to canonical testability.

Our first result relates between earthmover resilience, tolerant testability and canonical testability in images and edge colored ordered graphs. 

\begin{theorem}[See also Theorem \ref{thm:earthmover_tolerant_to_canonical}]
\label{thm:earthmover_tolerant_informal}
The following conditions are equivalent for any property $\mathcal{P}$ of edge colored ordered graphs or images. 

\begin{enumerate}
	\item $\mathcal{P}$ is earthmover resilient and tolerantly testable.
	\item $\mathcal{P}$ is canonically testable.
\end{enumerate}
\end{theorem}
Theorem \ref{thm:earthmover_tolerant_to_canonical}, which is the more detailed version of Theorem \ref{thm:earthmover_tolerant_informal}, also states that \emph{efficient} tolerant  $(\epsilon, \delta)$-tests -- in which the query complexity is polynomial in $\delta(\epsilon)$ -- can be converted, under certain conditions, into efficient canonical tests, and vice versa. 

Let us note that Theorem \ref{thm:earthmover_tolerant_informal} can be extended to high-dimensional ordered structures, such as tensors (e.g. 3D images) or edge colored ordered hypergraphs. As our focus in this paper is on one- and two-dimensional structures, the full proof of the extended statement is not given here, but it is a straightforward generalization of the $2D$ proof.

In (one-dimensional) strings, it turns out that the tolerant testability condition of Theorem \ref{thm:earthmover_tolerant_informal} is not needed. That is, ER and canonical testability are equivalent for string properties.
\begin{theorem}
	\label{thm:string_characterization}
	A string property $\mathcal{P}$ is canonically testable if and only if it is earthmover resilient.
\end{theorem}

In the unordered graph case, it was shown that testability is equivalent to estimability \cite{FischerNewman2005} and to regular reducibility \cite{AlonFischerNewmanShapira2009}. Here, we establish analogous results for canonical tests in ordered structures. The notion of (ordered) regular reducibility that we use here is similar in spirit to the unordered variant, but is slightly more involved. The formal definition is given in Subsection \ref{subsec:regular_reducible}.  

\begin{theorem}
	\label{thm:canonical_to_estimable}
	Any canonically testable property of edge colored ordered graphs and images is (canonically) estimable.
\end{theorem}

\begin{theorem}
	\label{thm:canonical_vs_regular_reducible}
	A property of edge colored ordered graphs or images is canonically testable if and only if it is regular reducible.
\end{theorem}

The characterization of tolerant testability in ER properties, given below, is a direct corollary of Theorems \ref{thm:earthmover_tolerant_informal}, \ref{thm:canonical_to_estimable}, and \ref{thm:canonical_vs_regular_reducible}. 
\begin{corollary}
The following conditions are equivalent for any earthmover resilient property $\mathcal{P}$ of edge colored ordered graphs or images.
\begin{enumerate}
	\item $\mathcal{P}$ is tolerantly testable.
	\item $\mathcal{P}$ is canonically testable.
	\item $\mathcal{P}$ is estimable.
	\item $\mathcal{P}$ is regular reducible.
\end{enumerate}
\end{corollary}
While the conversion between tolerant tests and canonical tests (and vice versa) among earthmover resilient properties has a reasonable polynomial blowup in the number of queries under certain conditions, for the relation between canonical testability and estimability or regular reducibility this is not known to be the case.
The proofs of Theorems \ref{thm:canonical_to_estimable} and \ref{thm:canonical_vs_regular_reducible} go through Szemer\'edi-regularity type arguments, and thus yields at least a tower-type blowup in the number of queries. Currently, it is not known how to avoid this tower-type blowup in general, even for unordered graphs. However, interesting recent results of Hoppen, Kohayakawa, Lang, Lefmann and Stagni \cite{HKLLS2016, HKLLS2017} state that for hereditary properties of unordered graphs, the blowup between testability and estimability is at most exponential.

Alon and the authors \cite{AlonBenEliezerFischer2017} recently showed that any hereditary property of edge-colored ordered graphs and images is canonically testable, by proving an order-preserving removal lemma for all such properties. From Theorem \ref{thm:canonical_to_estimable} and  \cite{AlonBenEliezerFischer2017} we derive the following very general result. 
\begin{corollary}
	\label{coro:hereditary_estimable}
	Any hereditary property of edge-colored ordered graphs or images is (canonically) estimable.
\end{corollary}
In particular, this re-proves the estimability of previously investigated properties such as monotonicity \cite{DodisGLRRS1999, ErgunKannanKRV2000} and more generally $k$-monotonicity \cite{CanonneGGKW2016}, and proves the estimability of forbidden-submatrix and forbidden-poset type properties \cite{AlonBenEliezer2017, AlonBenEliezerFischer2017, AlonFischerNewman2007, FischerNewman2007, FischerRozenberg2007}.

\begin{remark}
The characterization of the one-sided error obliviously testable properties by Alon and Shapira \cite{AlonShapira2008}, mentioned in Subsection \ref{subsec:characterization_unordered}, carries on to canonical tests in ordered graphs and images. That is, a property $\mathcal{P}$ of such structures has a one-sided error oblivious canonical test if and only if it is (essentially) hereditary. The fact that hereditary properties are obliviously canonically testable with one-sided error is proved in \cite{AlonBenEliezerFischer2017}; the proof of the other direction is very similar to its analogue in unordered graphs \cite{AlonShapira2008}, and is therefore omitted.
\end{remark}

\subsection{Related work}
\paragraph{Canonical versus sample-based testing in strings}
\label{subsec:intro_strings}
The notion of a sample-based test, already defined in the seminal work of Goldreich, Goldwasser and Ron \cite{GoldreichGoldwasserRon1998}, refers to tests that cannot choose which queries to make. 
A $q$-query test for $\mathcal{P}$ is \emph{sample-based} if it receives pairs of the form $(x_1, f(x_1)), \ldots, (x_q, f(x_q))$ where $f$ is the unknown input function and $x_1, \ldots, x_q$ are picked uniformly at random from the domain of $X$ (compare this to the definition of canonical tests from Subsection \ref{subsec:unordered_to_ordered}). 
A recent work of Blais and Yoshida \cite{BlaisYoshida2016} characterizes the properties $\mathcal{P}$ that have a constant query \emph{sample-based} test. 

In strings, sample-based testability might seem equivalent to \emph{canonical} testability at first glance, but this is actually not the case, as sample-based tests have more power than canonical ones (canonical testability implies sample-based testability, but the converse is not true). Consider, e.g., the property of equality to the string $010101\ldots$, which is trivially sample-based testable, yet not canonically testable.
Thus, sample-based testability does not imply canonical testability, so the results of Blais and Yoshida \cite{BlaisYoshida2016} are not directly comparable to Theorem \ref{thm:string_characterization} above.

\paragraph{Previously investigated properties of ordered structures}
On top of the hereditary properties mentioned earlier, several different types of properties of ordered structures have been investigated in the property testing literature. Without trying to be comprehensive, here is a short summary of some of these types of properties.

	\subparagraph{Geometric \& visual properties} Image properties that exhibit natural visual conditions, such as connectivity, convexity and being a half plane, were considered e.g. in \cite{BermanMR2016,  BermanMurzabulatovR2016, CFSS2017, Raskhodnikova2003}. Typically in these cases, images with two colors -- black and white -- are considered, where the ``shape'' consists of all black pixels, and the ``background'' consists of all white pixels. For example, convexity simply means that the black shape is convex. As we shall see, some of these properties that are global in nature, such as convexity and being a half plane, are ER, while connectivity -- a property that is sensitive to local modifications -- is not ER. 
	
	\subparagraph{Algebraic properties} String properties related to low-degree polynomials, PCPs and locally testable error correcting codes have been thoroughly investigated, starting with the seminal papers of Rubinfeld and Sudan \cite{RubinfeldSudan1996} and Goldreich and Sudan \cite{GoldreichSudan2002}. As shown in \cite{FischerFortnow2006}, there exist properties of this type that are testable but not tolerantly testable. In this sense, algebraic properties behave very differently from unordered graph properties. This should not come as a surprise: In a PCP or a code, the exact \emph{location} of each bit is majorly influential on its ``role''. This kind of properties is therefore not ER in general.
	
	\subparagraph{Local properties} These are image properties $\mathcal{P}$ where one can completely determine whether a given image $\mathcal{I}$ satisfies $\mathcal{P}$ based only on the statistics of the $k \times k$ consecutive sub-images of $\mathcal{I}$, for a fixed constant $k$. 
	Recently, Ben-Eliezer, Korman and Reichman \cite{BenEliezerKormanReichman2017} observed that for almost all (large enough) patterns $Q$, the local property of not containing a \emph{consecutive} copy of $Q$ in the image is tolerantly testable. Note that monotonicity can also be represented as a local property, taking $k=2$ (but $\ell$-monotonicity cannot be represented this way). Local properties are not ER in general, and obtaining characterizations of testability for them remains an intriguing open problem.

\section{Preliminaries}
\label{sec:prelims}

This Section contains all required definitions, including those that are related to earthmover resilience (Subsection \ref{subsec:earthmover_resilience}), a discussion on earthmover resilient properties (Subsection \ref{subsec:earthmover_visual}), property testing notation (Subsection \ref{subsec:testing_defs}), and finally, the definition of ordered regular reducibility (Subsection \ref{subsec:regular_reducible}).
Along the way, we state the full version of Theorem \ref{thm:earthmover_tolerant_informal} (Subsection \ref{subsec:full_version_statement}).

We start with some standard definitions.
A \emph{property} $\mathcal{P}$ of functions $f \colon X \to \Sigma$ is simply viewed as a collection of such functions, where $f$ is said to \emph{satisfy} $\mathcal{P}$ if $f \in \mathcal{P}$. 
The \emph{absolute} Hamming distance between two functions $f,f':X \to Y$ is $D_H(f,f') = |\{ x \in X : f(x) \neq f'(x) \}|$, and the \emph{relative distance} is $d_H(f, f') = D_H(f, f') / |X|$; note that $0 \leq d_H(f, f') \leq 1$ always holds. $f$ and $f'$ are \emph{$\epsilon$-far} if $d_H(f,f') > \epsilon$, and \emph{$\epsilon$-close} otherwise. The distance of $f$ to a property $\mathcal{P}$ is $\min_{f' \in \mathcal{P}} d_H(f,f')$. $f$ is \emph{$\epsilon$-far} from $\mathcal{P}$ if the distance between $f$ and $\mathcal{P}$ is larger than $\epsilon$, and \emph{$\epsilon$-close} to $\mathcal{P}$ otherwise.

\paragraph{Representing images using ordered graphs}
An image $f \colon [n] \times [n] \to \Sigma$ can be represented by an edge colored ordered graph $g \colon \binom{[2n]}{2} \to \Sigma \cup \{\sigma\}$, where $\sigma \notin \Sigma$ can be thought of as a special ``no edge'' symbol. $g$ is defined as follows. $g(x,y) = \sigma$ for any pair $x \neq y$ satisfying $1 \leq x,y \leq n$ (``pair of rows'') or $n+1 \leq x,y \leq 2y$ (``pair of columns''); and $g(x,n+y) = f(x,y)$ for any $x,y \in [n]$. From now onwards, we almost exclusively use this representation of images as ordered graphs, usually giving our definitions and proofs only for strings and ordered graphs. It is not hard to verify that all results established for ordered graphs can be translated to images through this representation.

\subsection{Earthmover resilience}
\label{subsec:earthmover_resilience}
We now formalize our notion of being ``well behaved''. As both strings and ordered graphs are essentially functions of the form $f \colon \binom{[n]}{k} \to \Sigma$ (for $k=1$ and $k=2$, respectively), we simplify the presentation by giving here the general definition for functions of this type.
\begin{definition}[Earthmover distance]
	\label{def:earth_distance}
	Fix $k > 0$ and let $f \colon \binom{[n]}{k} \to \Sigma$. A \emph{basic move} between consecutive elements $x,x+1 \in [n]$ in $f$ is the operation of swapping $x$ and $x+1$ in $f$. Formally, let $\sigma_x \colon [n] \to [n]$ be the permutation satisfying $\sigma_x(x) = x+1$, $\sigma_x(x+1) = x$, and $\sigma_x(i) = i$ for any $i \neq x,x+1$. For any $X \in \binom{[n]}{k}$, define $\sigma^k_x(X) = \{\sigma_x(i) \colon i \in X\}$.
	The result of a basic move between $x$ and $x+1$ in $f$ is the composition $f' = f \circ \sigma_x^k$.
	
	The {\em absolute earthmover distance} $D_e(f,f')$ between two functions $f, f' \colon \binom{[n]}{k} \to \Sigma$ is the minimum number of basic move operations needed to produce $f'$ from $f$. The distance is defined to be $+\infty$ if $f'$ cannot be obtained from $f$ using any number of basic moves.
	The \emph{normalized earthmover distance} between $f$ and $f'$ is $d_e(f,f') = D_e(f,f') / \binom{n}{2}$, and we say that they are $\epsilon$-earthmover-far if $d_e(f,f') > \epsilon$, and $\epsilon$-earthmover-close otherwise.
\end{definition}

\begin{definition}[Earthmover resilience]
	Fix a function $\delta \colon (0,1) \to (0,1)$. A property $\mathcal{P}$ is \emph{$\delta$-earthmover resilient} if for any $\epsilon > 0$, function $f$ satisfying $\mathcal{P}$, and function $f'$ which is $\delta(\epsilon)$-earthmover-close to $f$, it holds that $f'$ is $\epsilon$-close to $\mathcal{P}$ (in the usual Hamming distance).
	$\mathcal{P}$ is \emph{earthmover resilient} if it is $\delta$-earthmover resilient for some choice of $\delta$.
\end{definition}
Intuitively, a property is earthmover resilient if it is insensitive to local changes in the order of the base elements. 

\paragraph{Hereditary properties are earthmover resilient}
It was shown in \cite{AlonBenEliezerFischer2017} that any hereditary property satisfies a \emph{removal lemma}: If an ordered graph (or image) $G$ is $\epsilon$-far from an hereditary property $\mathcal{P}$, then $G$ contains $\delta n^h$ ordered copies of some $h$-vertex subgraph $H$ not satisfying $\mathcal{P}$, for suitable choices of $\delta = \delta_{\mathcal{P}}(\epsilon) > 0$ and $h = h_{\mathcal{P}}(\epsilon) > 0$. Since one basic move can destroy no more than $n^{h-2}$ such $H$-copies (those that include both swapped vertices), one has to make at least $\delta n^2$ basic moves to make $G$ satisfy $\mathcal{P}$. Thus, $\epsilon$-farness implies $\delta_{\mathcal{P}}(\epsilon)$-earthmover-farness from $\mathcal{P}$.

\subsection{Earthmover resilience in visual properties}
\label{subsec:earthmover_visual}
Convexity and being a half plane are earthmover resilient. This is a special case of a much wider phenomenon concerning properties of black-white images in which the number of pixels lying in the boundary between the black shape and the white background is small. Here, an $m \times n$ white/black image is represented by a $0/1$-matrix $M$ of the same dimensions, where the $(i,j)$-pixel of the image is black if and only if $M(i,j) = 1$.
The definition below is given for square images, but can be easily generalized to $m \times n$ images with $m = \Theta(n)$.
\begin{definition}[Sparse boundary]
	\label{def:sparse_boundary}
	The \emph{boundary} $\mathcal{B} = \mathcal{B}(\mathcal{I})$ of an $n \times n$ black-white image $\mathcal{I}$ is the set of all pixels in $\mathcal{I}$ that are black and have a white neighbor.\footnote{Here, two pixels are neighbors if they share one coordinate and differ by one in the other coordinate. An alternative definition (that will yield the same results in our case) is that two pixels are neighbors if they differ by at most one in each of the coordinates, and are not equal.} $\mathcal{B}$ is \emph{$c$-sparse} for a constant $c > 0$ if $|\mathcal{B}| \leq c n$. A property $\mathcal{P}$ has a \emph{$c$-sparse boundary} if the boundaries of all images satisfying $\mathcal{P}$ are $c$-sparse.
\end{definition}

For example, for any property $\mathcal{P}$ of $n \times n$ images such that the black area in any image satisfying $\mathcal{P}$ is the union of at most $t$ convex shapes (that do not have to be disjoint), $\mathcal{P}$ has a $4t$-sparse boundary. This follows from the fact that the boundary of each of the black shapes is of size at most $4n$. For $t = 1$, this captures both convexity and being a half plane as special cases. The following result states that $c$-sparse properties are earthmover resilient.

\begin{theorem}
	\label{thm:sparse_boundary_to_earthmover_resilience}
	Fix $c \geq 1$. Then any property with a $c$-sparse boundary is $\delta$-earthmover-resilient, where $\delta(\epsilon) \leq \alpha \epsilon^2 / c^2$ for some absolute constant $\alpha > 0$ and any $\epsilon > 0$.
\end{theorem}
The result still holds if $c$ is taken as a function of $\epsilon$. The (non-trivial) proof serves as a good example showing how to prove earthmover resilience of properties, and is given in Appendix \ref{app:sparse_boundary}. 

Naturally, not all properties of interest are earthmover resilient. For example, the local property $\mathcal{P}$ of ``not containing two consecutive horizontal black pixels'' in a black/white image is not earthmover resilient: Consider the chessboard $n \times n$ image, which satisfies $\mathcal{P}$, but by partitioning the board into $n/4$ quadruples of consecutive columns and switching between the second and the third column in each quadruple, we get an image that is $O(1/n)$-earthmover-close to $\mathcal{P}$ yet $1/4$-far from it in Hamming distance.
A similar but slightly more complicated example shows that connectivity is not earthmover resilient as well.

\subsection{Definitions: Testing and estimation}
\label{subsec:testing_defs}
A $q$-query algorithm $T$ is said to be an \emph{$\epsilon$-test} for $\mathcal{P}$ with \emph{confidence} $c > 1/2$, if it acts as follows. Given an unknown input function $f:X \to \Sigma$ (where $X$ and $\Sigma$ are known), $T$ picks $q$ elements $x_1, \ldots, x_q \in X$ of its choice, and queries the values $f(x_1), \ldots, f(x_q)$.\footnote{$T$ as defined here is a \emph{non-adaptive} test, that chooses which queries to make in advance. Adaptivity does not matter for our discussion, since we are only interested in constant-query tests, and since an adaptive test making a constant number $q$ of queries can be turned into a non-adaptive one making $2^q$ queries, which is still a constant.}  
Then $T$ decides whether to accept or reject $f$, so that 
\begin{itemize}
	\item If $f$ satisfies $\mathcal{P}$ then $T$ accepts $f$ with probability at least $c$. 
	\item If $f$ is $\epsilon$-far from $\mathcal{P}$, then $T$ rejects it with probability at least $c$. 
\end{itemize}
Now let $\delta:(0,1) \to (0,1)$ be a function that satisfies $\delta(x) < x$ for any $0 < x < 1$. An \emph{$(\epsilon, \delta)$-tolerant} test $T$ is defined similarly to an $\epsilon$-test, with the first condition replaced with the following strengthening: If $f$ is \emph{$\delta(\epsilon)$-close} to $\mathcal{P}$, then $T$ accepts it with probability at least $1-c$.
Unless stated otherwise, the default choice for the confidence is $c=2/3$.
$\mathcal{P}$ is \emph{testable} if it has a constant-query $\epsilon$-test (whose number of queries depends only on $\epsilon$) for any $\epsilon > 0$. Similarly, $\mathcal{P}$ is $\delta$-tolerantly testable, for a valid choice of $\delta \colon (0,1) \to (0,1)$, if it has a constant query $(\epsilon, \delta)$-test for any $\epsilon > 0$. If $\mathcal{P}$ is $\delta$-tolerantly testable for \emph{some} valid choice of $\delta$, we say that it is \emph{tolerantly testable}. Finally, $\mathcal{P}$ is \emph{estimable} if it is $\delta$-tolerantly testable for \emph{any} valid choice of $\delta$.

Next, we formally define what it means for a test (or a tolerant test) $T$ to be \emph{canonical}, starting with the definition for strings.

A $q$-query test (or tolerant test) $T$ for a property $\mathcal{P}$ of strings $f:[n] \to \Sigma$ is \emph{canonical} if it acts in two steps. First, it picks $x_1 < \ldots < x_q$ uniformly at random, and queries the entries $y_1 = f(x_1), \ldots, y_q = f(x_q)$. The second step only receives the ordered tuple $Y = (y_1, \ldots, y_k)$ and decides  (possibly probabilistically) whether to accept or reject only based on the values of $Y$. Note that the second step does \emph{not} ``know'' the values of $x_1, \ldots, x_q$ themselves. As before, $\mathcal{P}$ is \emph{canonically testable} if it has a $q_{\mathcal{P}}(\epsilon)$-query canonical test for any $\epsilon > 0$, where $q_{\mathcal{P}}(\epsilon)$ depends only on $\epsilon$.

In contrast, a test for string properties is \emph{sample based} if it has the exact same first step, but the second step receives more information: It also receives the values of $x_1, \ldots, x_q$. 
A sample-based test is more powerful than a canonical test  in general. For example, the property of ``being equal to the string $010101\ldots$'' is trivially sample-based $\epsilon$-testable with $O(1/\epsilon)$ queries, but is \emph{not} canonically testable with a constant number of queries (that depends only on $\epsilon$).

For \emph{ordered graphs} $f\colon \binom{[n]}{2} \to \Sigma$, a test (or a tolerant test) $T$ is \emph{canonical} if, again, it acts in two steps. In the first step, $T$ picks $q$ vertices $v_1 < \ldots < v_q$ uniformly at random, and queries all $\binom{q}{2}$ values $y_{ij} = f(v_i, v_j)$. The second step receives the ordered tuple $Y = (y_{11}, y_{12}, \ldots, y_{1q}, \ldots, y_{q-1,q})$, and decides (possibly probabilistically) whether to accept or reject only based on the value of $Y$. 

We take a short detour to explain why asking $T$ to make a deterministic decision in the second step of the canonical test, rather than a probabilistic one, will not make an essential difference for our purposes.  It was proved by Goldreich and Trevisan \cite{GoldreichTrevisan2003} that any probabilistic canonical test (for which the decision to accept or reject in the second step is not necessarily deterministic) can be converted into a deterministic one, with a blowup that is at most polynomial in the number of queries.
The proof was given for unordered graph properties, but it can be translated to ordered structures like strings, ordered graphs and images in a straightforward manner. Thus, the requirement that the canonical test makes a deterministic decision is not restrictive. 

\subsection{The full statement of Theorem \ref{thm:earthmover_tolerant_informal}}
\label{subsec:full_version_statement}
We are finally ready to present the more precise version of Theorem \ref{thm:earthmover_tolerant_informal}. This version depicts an \emph{efficient} transformation from earthmover resilience and tolerant testability to canonical testability, and vice versa.
\begin{theorem}
	\label{thm:earthmover_tolerant_to_canonical}
	Let $\mathcal{P}$ be a property of edge-colored ordered graphs or images, and let $\delta \colon (0,1) \to (0,1)$ and $\eta \colon (0,1) \to (0,1)$ such that $\eta(\epsilon) < \epsilon$ for any $\epsilon > 0$.
	\begin{enumerate}
		\item If $\mathcal{P}$ is $\delta$-earthmover resilient and $\eta$-tolerantly testable, where the number of queries of a corresponding $(\epsilon, \eta)$-tolerant non-adaptive test is denoted by $q(\epsilon)$, then $\mathcal{P}$ is canonically testable. Moreover, if $q$, $\eta^{-1}$ and $\delta^{-1}$ are polynomial in $\epsilon^{-1}$, then the number of queries of the canonical $\epsilon$-test is also polynomial in $\epsilon^{-1}$.
		
		\item If $\mathcal{P}$ is canonically testable, where the number of queries of the canonical (non adaptive) $\epsilon$-test is denoted by $q'(\epsilon)$, then $\mathcal{P}$ is both $\delta'$-earthmover resilient and $\delta'$-tolerantly testable where $\delta' \colon (0,1) \to (0,1)$ depends only on $q'$ and $\epsilon$. Moreover, if $q'$ is polynomial in $\epsilon^{-1}$, then $\delta'$ is polynomial in $\epsilon$.
	\end{enumerate}
\end{theorem}
The proof is given along Sections \ref{sec:earthmover_mixing}, \ref{sec:earth_tolerant_to_piecewise}, and \ref{sec:piecewise_to_canonical}.

\subsection{Regular reducibility}
\label{subsec:regular_reducible}
The last notion to be formally defined is that of ordered regular reducibility. This notion is a natural analogue of the unordered variant, and is rather complicated to describe and define. Since the intuition behind this definition is quite similar to that of the unordered case, we refer the reader to a more thorough discussion on regular reducibility (and the relation to Szemer\'edi's regularity lemma) in \cite{AlonFischerNewmanShapira2009}. Here, we only provide the set of definitions required for our purposes.

\begin{definition}[Regularity,  regular partition]
	Let $f \colon \binom{[n]}{2} \to \Sigma$ be an edge-colored ordered graph. For any $\sigma \in \Sigma$, the \emph{$\sigma$-density} of a disjoint pair $A,B \subseteq [n]$ is $d_{\sigma}(A,B) = |f^{-1}(\sigma) \cap (A \times B)| / |A||B|$. A pair $(A, B)$ is $\gamma$-regular if for any two subsets $A' \subseteq A$ and $B' \subseteq B$ satisfying $|A'| \geq \gamma |A|$ and $|B'| \geq \gamma |B|$, and any $\sigma \in \Sigma$, it holds that $|d_{\sigma}(A',B') - d_{\sigma}(A,B)| \leq \gamma$.
	An equipartition of $[n]$ into $k$ parts $V_1, \ldots, V_k$ is \emph{$\gamma$-regular} if all but at most $\gamma \binom{k}{2}$ of the pairs $(V_i, V_j)$ are $\gamma$-regular.
\end{definition}
\begin{definition}[Interval partitions]
	The \emph{$k$-interval equipartition} of $[n]$ is the unique partition of $[n]$ into sets $X_1, \ldots, X_k$, such that $x < x'$ for any $x \in X_i, x' \in X_{i'}, i < i'$ and $|X_{i'}| \leq |X_i| \leq |X_{i'}|+1$ for any $i < i'$. An \emph{interval partition} of an ordered graph or a string is defined similarly. 
\end{definition}
\begin{definition}[Ordered regularity instance]
	\label{def:ordered_regularity_instance}
	An ordered regularity instance $R$ for $\Sigma$-colored ordered graphs is given by an error parameter $\gamma$, integers $r,k$, a set of $K = \binom{r}{2}k^2 |\Sigma|$ densities $0 \leq \eta_{ij}^{i'j'}(\sigma) \leq 1$ indexed by $i < i' \in [r]$, $j,j' \in [k]$ and $\sigma \in \Sigma$, and a set $\bar{R}$ of tuples $(i,j,i',j')$ of size at most $\gamma K$. An ordered graph $f \colon \binom{[n]}{2} \to \Sigma$ satisfies the regularity instance if there is an equitable refinement $\{V_{ij} : i \in [r], j \in [k]\}$ of the $r$-interval equipartition $V_1, \ldots, V_r$ where $V_{ij} \subseteq V_i$ for any $i$ and $j$,  such that for all $(i,j,i',j') \notin \bar R$ the pair $V_{ij}, V_{i'j'}$ is $\gamma$-regular and satisfies $d_{\sigma}(V_{ij}, V_{i'j'}) = \eta_{ij}^{i'j'}(\sigma)$ for any $\sigma \in \Sigma$. The \emph{complexity} of the regularity instance is $\max \{1/\gamma, K\}$.
\end{definition}
With some abuse of notation, when writing $d_{\sigma}(V_{ij}, V_{i'j'}) = \eta_{ij}^{i'j'}(\sigma)$ we mean that the number of $\sigma$-colored edges between $V_{ij}$ and $V_{i'j'}$ is $\lfloor \eta_{ij}^{i'j'}(\sigma)|V_{ij}| |V_{i'j'}|  \rfloor$ or $\lceil \eta_{ij}^{i'j'}(\sigma)|V_{ij}| |V_{i'j'}|  \rceil$. This way we avoid divisibility issues, without affecting any of our arguments.

The definition of an ordered regularity instance differs slightly from the analogous definition for unordered graphs in \cite{AlonFischerNewmanShapira2009}: Here we insist that the regular partition will be a refinement of an interval equipartition, disregarding pairs of parts inside the same interval. We also allow a color set of size bigger than two.
The definition of regular reducibility is analogous to the unordered case, though obviously the regularity instances used in the definition are of the ordered type.
\begin{definition}[Regular reducible]
	\label{def:ordered_regular_reducible}
	An edge-colored ordered graph property $\mathcal{P}$ is \emph{regular-reducible} if for any $\delta > 0$ there exists $t = t_{\mathcal{P}}(\delta)$ such that for any $n$ there is a family $\mathcal{R}$ of at most $t$ regularity instances, each of complexity at most $t$, such that the following holds for every $\epsilon > 0$ and ordered graph $f \colon \binom{[n]}{2} \to \Sigma$: 
	\begin{itemize}
		\item If $f$ satisfies $\mathcal{P}$ then for some $R \in\mathcal{R}$, $f$ is $\delta$-close to satisfying $R$.
		\item If $f$ is $\epsilon$-far from satisfying $\mathcal{P}$, then for any $R \in \mathcal{R}$, $f$ is $(\epsilon - \delta)$-far from satisfying $R$.
	\end{itemize}
	
\end{definition}

\section{Proof outline}
\label{sec:proof_outline}

	In this section, we shortly describe the main ingredients of our proofs.
	
	\paragraph{Earthmover distance and mixingness}
	Suppose that $G,G':\binom{[n]}{2} \to \Sigma$ are two ordered graphs with a finite earthmover distance between them (all results mentioned here also apply for strings). In this case, $G$ and $G'$ are isomorphic as unordered graphs, meaning that the collection of vertex permutations $\pi:[n] \to [n]$ that ``turn'' $G$ into $G'$ is not empty. We define the (absolute) \emph{mixingness} between $G$ and $G'$ as the minimal number of pairs $x < y \in [n]$ such that $\pi(x) > \pi(y)$, over all possible choice of $\pi$ from the collection. We show, via a simple inductive proof, that the mixingness between $G$ and $G'$ is \emph{exactly equal} to the earthmover distance between them.
	
	With the tool of mixingness in hand, it is not hard to prove that canonical testability implies earthmover resilience and tolerant testability. The basic idea is that, if two graphs $G$ and $G'$ are sufficiently close in terms of mixingness, then the distributions of their $q$-vertex subgraphs are very similar, and so a $q$-query canonical test cannot distinguish between them with good probability. See Section \ref{sec:earthmover_mixing} for more details.
	
	\paragraph{Earthmover resilience to piecewise-canonical testability}
	A test $T$ is \emph{piecewise-canonical} if it acts in the following manner on the $t$-interval partition of the unknown input graph (or string). First, $T$ chooses how many vertices (entries, respectively) to take from each interval, where the number of vertices may differ between different intervals. Then $T$ picks the vertices (entries) from the intervals in a uniformly random manner. 
	Finally, $T$ queries precisely all pairs of picked vertices (or all entries, in the string case), and decides whether to accept or reject based on the ordered tuple of the values returned by the queries. 
	
	For strings of length $n$ over $\Sigma$, if $\mathcal{P}$ is earthmover resilient then it is also piecewise-canonically testable. The main idea of the proof is the following. 
	If one takes a string $S$ and partitions it into sufficiently many equitable interval parts $S_1, \ldots, S_t$, then ``shuffling'' entries inside each of the interval parts $S_i$ will not change the distance of $S$ to $\mathcal{P}$ significantly.
	With this idea in hand, it is not hard to observe that knowing the histograms $H_i$ of all parts $S_i$ (with respect to letters in $\Sigma$) is enough to estimate the distance of $S$ to $\mathcal{P}$ up to a small additive constant error.
	These histograms cannot be computed exactly with a constant number of queries, but it is well known that each $H_i$ can be estimated up to a small constant error with a constant number of queries, which is enough for our purposes.
	
	For properties $\mathcal{P}$ of ordered graphs (or images), earthmover resilience by itself is not enough to imply piecewise-canonical testability, but earthmover resilience and tolerant testability are already enough. The idea is somewhat similar to the one we used for strings.
	We may assume that $\mathcal{P}$ has a tolerant test $T$ whose set of queried pairs is always an induced subgraph of $G$.
	Like before, we partition our input graph $G$ into sufficiently many interval parts $V_1, \ldots, V_t$. Now the piecewise canonical test $T^*$ simulates a run of the original tolerant test $T$ (without making the actual queries that $T$ decided on). Denote the vertices that $T$ decides to pick in $V_i$ by $v_1^i, \ldots, v_{q_i}^{i}$. $T^*$ picks exactly $q_i$ vertices uniformly at random in each part $V_i$, and queries all edges between all chosen vertices. Now $T^{*}$ randomly ``assigns'' the labels $v_1^i, \ldots, v_{q_i}^{i}$ to the vertices that it queried from $V_i$, and returns the same answer that $T$ would have returned for this set of queries.
	It can be shown that $T^*$ is a test whose probability to return the same answer as $T$ is high, as desired. 
	For the full details, see Section \ref{sec:earth_tolerant_to_piecewise}.
	
	\paragraph{Piecewise-canonical testability to canonical testability}
	We describe the transformation for ordered graph properties; for strings this is very similar.
	Let $T$ be piecewise-canonical test for $\mathcal{P}$ that partitions the input into $t$ intervals $U_1, \ldots, U_t$.
	Consider the following canonical test $T'$: $T'$ picks $qt$ vertices $v_1 < \ldots < v_{qt}$ uniformly at random, for large enough $q$. Then $T'$ partitions the vertices into $t$ intervals $A_1 = \{v_1, \ldots, v_q\}, \ldots, A_t = \{v_{(t-1)q + 1}, v_{tq}\}$. Now $T'$ simulates a run of $T$. If $T$ chose to take $q_i$ vertices from $U_i$, then $T'$ picks exactly $q_i$ vertices from $A_i$. Finally, $T'$ queries all edges between all vertices it picked, and returns the same answer as $T$ (where the simulation of $T$ assumes here that the vertices that were actually picked from $A_i$ come from $U_i$).
	
	A rather straightforward but somewhat technical proof (that we do not describe at this point, see Section \ref{sec:piecewise_to_canonical}) shows that the probability that $T'$ returns the answer that $T$ would have returned on the same input is high, establishing the validity of $T'$.
	For the full details, see Section \ref{sec:piecewise_to_canonical}.
	
	\paragraph{Canonical testability, estimability and regular reducibility}
	The proofs of Theorems \ref{thm:canonical_to_estimable} and \ref{thm:canonical_vs_regular_reducible} are technically involved. Fortunately, the proofs follow the same spirit as those of the unordered case, considered in \cite{AlonFischerNewmanShapira2009, FischerNewman2005}, and in this paper we only describe how to adapt the unordered proofs to our case. 	
	
	Sections \ref{sec:canonical_to_estimable} and \ref{sec:canonical_vs_regular_reducibility} contain the proofs of Theorems \ref{thm:canonical_to_estimable} and \ref{thm:canonical_vs_regular_reducible}, respectively. It is shown in these sections that for our ordered case, in some sense it is enough to make the proofs work for $k$-partite graphs, for a fixed $k$. The intuition is that for our purposes, it is enough to view an ordered graph $G$ as a $k$-partite graph (for a large enough constant $k$), where the parts are the intervals of a $k$-interval partition of $G$.

	At this point, it is too difficult to explain the proof idea in high level without delving deeply into the technical details. Therefore, all details are deferred to Sections \ref{sec:canonical_to_estimable} and \ref{sec:canonical_vs_regular_reducibility}.

\section{Discussion and open problems}
\label{sec:discussion}
The earthmover resilient properties showcase, among other phenomena, an interesting connection between visual properties of images and the regularity-based machinery that was previously used to investigate unordered graphs. 
We believe that further research on the characterization problem for ordered structures would be interesting. It might also be interesting to investigate such problems using distance functions that are not Hamming distance, as was done, e.g., in \cite{BermanRaskhodnikovaYaroslavtev2014}.
Finally we present two open questions.

\paragraph{Characterization of testable earthmover-resilient properties}
In this work we provide a characterization of earthmover resilient tolerantly testable properties. Although using such tests might make more sense than using intolerant tests in the presence of noise in the input (a situation that is common in areas like image processing, that are related to image property testing), it would also be very interesting to provide a characterization of the \emph{testable} earthmover resilient properties. In particular, does there exist an earthmover resilient property that is testable but not tolerantly testable? The only known example of a (non earthmover resilient) property that is testable but not tolerantly testable is the PCPP-based property of \cite{FischerFortnow2006}, and it will certainly be interesting to find more examples of properties that have this type of behavior.

\paragraph{Alternative classes of properties} 
The class of earthmover resilient properties captures properties that are global in nature, and it will be interesting to identify and analyze some other wide classes of properties. A natural candidate is the class of all local properties \cite{BenEliezerKormanReichman2017}. We also believe that it might be possible to find other interesting classes of visual properties.

\section{Earthmover-resilience and mixing}
\label{sec:earthmover_mixing}
\begin{definition}
	Let $\mu$ and $\eta$ be two distributions over a finite family $\mathcal{H}$ of combinatorial structures. The \emph{variation distance} between $\mu$ and $\eta$ is $|\mu - \eta| = \frac{1}{2} \sum_{H \in \mathcal{H}} |\Pr_{\mu}(H) - \Pr_{\eta}(H)|$.
\end{definition}
The following folklore fact regarding the variation distance will be useful later. 
\begin{lemma}
	\label{lem:variation1}
	Let $\mu$ and $\eta$ be two distributions over a finite family $\mathcal{H}$. Then $|\mu - \eta| = \max_{\mathcal{F} \subseteq \mathcal{H}} |\Pr_{\mu}(\mathcal{F}) - \Pr_{\eta}(\mathcal{F})| = \sum_{H \in \mathcal{H}:\ \Pr_{\mu}(H) > \Pr_{\eta}(H)} (\Pr_{\mu}(H) - \Pr_{\eta}(H))$.
\end{lemma}

\begin{definition}
	\label{def:unordered_iso}
	An \emph{unordered isomorphism} between two ordered graphs $G,H:\binom{[n]}{2} \to \Sigma$ is a permutation $\sigma \colon[n] \to [n]$ such that $G(ij) = H(\sigma(i)\sigma(j))$ for any $i < j \in [n]$.
		
	Given a permutation $\sigma$ of $[n]$, the \emph{mixing set} of $\sigma$ is $MS(\sigma) = \{i<j:\sigma(i)>\sigma(j)\} \subseteq \binom{[n]}{2}$, its \emph{mixingness} is $D_m(\sigma) = |MS(\sigma)|$ and its
	{\em normalized mixingness} is $d_m(\sigma) = |MS(\sigma)|/\binom{n}{2}$. Given graphs $G$ and $H$, their normalized mixingness $d_m(G, H)$ is defined as the minimal normalized mixingness of an unordered isomorphism from $G$ to $H$ (and $+\infty$ if $G$ and $H$ are not isomorphic as unordered graphs).
\end{definition}

Our next goal is to show that the earthmover distance between two ordered graphs is equal to the mixingness between them. 
Given a permutation $\sigma \colon [n] \to [n]$, a \emph{basic move} for $\sigma$ transforms it to a permutation $\sigma'$ of the same length, such that for some $i$, $\sigma(i) = \sigma'(i+1)$ and $\sigma'(i) = \sigma(i+1)$, and $\sigma(j) = \sigma'(j)$ for any $j \neq i,i+1$.
Let $b(\sigma)$ denote the minimal number of basic moves required to turn $\sigma$ into the identity permutation $id$ satisfying $id(i) = i$ for any $i$.

\begin{lemma}
\label{lem:mixing_basic_perms}
$D_m(\sigma) = b(\sigma)$ for any permutation $\sigma \colon [n] \to [n]$.
\end{lemma}
\begin{proof}
The inequality $D_m(\sigma) \leq b(\sigma)$ is trivial: Any basic move changes the relative order between a (single) pair of entries in the permutation, and thus cannot decrease the size of the mixing set by more than one.
Next we show by induction that $b(\sigma) \leq D_m(\sigma)$. $D_m(\sigma) = 0$ implies that $\sigma = id$ and $b(\sigma) = 0$ in this case. Now assume that $D_m(\sigma) > 0$ and pick some $i < j$ such that $\sigma(i) > \sigma(j)$. Take $i' < j$ to be the largest for which $\sigma(i') > \sigma(j)$ -- such an $i'$ exists since $\sigma(i) > \sigma(j)$. Note that $\sigma(i'+1) \leq \sigma(j) < \sigma(i')$ due to the maximality of $i'$. Take $\sigma'$ to be the result of the basic move between $i'$ and $i'+1$ in $\sigma$. $D_m(\sigma') = D_m(\sigma)-1$, and by the induction assumption we know that $b(\sigma') = D_m(\sigma') = D_m(\sigma) - 1$. But since $\sigma'$ is the result of a basic move on $\sigma$, we conclude that $b(\sigma) \leq b(\sigma')+1 = D_m(\sigma)$, as desired.
\end{proof}

The equivalence between the earthmover distance and the mixingness is now immediate.
\begin{lemma}\label{lem:mixdists}
	For any two graphs $G, H \colon \binom{[n]}{2} \to \Sigma$, $d_e(G, H) = d_m(G,H)$.
\end{lemma}

\begin{proof}
	$D_m(G,H)$ is the minimum value of $D_m(\sigma)$ among all unordered isomorphisms $\sigma$ from $G$ to $H$, and $D_e(G,H)$ is the minimum value of $b(\sigma)$ among all such isomorphisms. By Lemma \ref{lem:mixing_basic_perms}, these two values are equal, and thus the corresponding relative measures are also equal.
\end{proof}

\begin{lemma}
	\label{lem:earthmover_to_edit}
	Let $\delta:(0,1) \to (0,1)$ and let $\mathcal{P}$ be a $\delta$-earthmover-resilient property. If two graphs $G,H \colon \binom{[n]}{2} \to \Sigma$ satisfy $d_e(G,H) \leq \delta(\epsilon)$ for some $\epsilon > 0$, then $d_H(G,\mathcal{P}) \leq d_H(H, \mathcal{P}) + \epsilon$.
\end{lemma}
\begin{proof}
	Suppose that $G$ and $H$ satisfy $d_e(G,H) \leq \delta(\epsilon)$. By definition, there exists an unordered isomorphism $\sigma \colon G \to H$ such that $d_m(G,H) = d_m(\sigma)$. Let $G' \colon \binom{[n]}{2} \to \Sigma$ be the graph in $\mathcal{P}$ that is closest to $G$ (in Hamming distance). 
	Consider the graph $H'$ satisfying $H'(\sigma(u)\sigma(v)) = G'(uv)$ for any $u \neq v \in V$, then $d_H(H, H') = d_H(G, G')$. Note that $\sigma$ is an unordered isomorphism between $G'$ and $H'$. It follows, building on Lemma \ref{lem:mixdists}, that
	$d_m(G', H') \leq d_m(\sigma) = d_m(G,H) = d_e(G,H) \leq \delta(\epsilon)$. This implies (by the earthmover resilience) that $H'$ is $\epsilon$-close to $\mathcal{P}$. The triangle inequality concludes the proof.
\end{proof}

\subsection*{Canonical testability implies earthmover resilience}

\begin{definition}
	Let $H$ and $G$ be $\Sigma$-edge-colored ordered graphs on $q$ and $n$ vertices respectively. The number of (ordered) copies of $H$ in $T$, i.e., the number of induced subgraphs of $G$ of size $q$ isomorphic to $H$, is denoted by $h(H, G)$. The \emph{density} of $H$ in $G$ is $t(H,G) = h(H, G) / \binom{n}{q}$ (where $t(H,G) = 0$ if $q > n$).
	The \emph{$q$-statistic} of $G$ is the vector $(t(H, G))_{H \in \mathcal{H}_q}$, where $\mathcal{H}_q$ is the family of all $\Sigma$-edge-colored ordered graphs with $q$ vertices.
\end{definition}

Every property of ordered graphs already testable by a canonical test is $\delta$-earthmover-resilient for some $\delta$ (depending on the number of its query vertices as a function of $\epsilon$), as implied by the following lemma.

\begin{lemma}\label{lem:smallcaninf}
	Let $\epsilon, \delta > 0$.
	For any canonical $\epsilon$-test querying up to $q$ vertices and any two graphs $G$ and $G'$ of either Hamming distance or earthmover distance at most $\delta$, the difference between the acceptance probabilities of $G$ and of $G'$ is at most $\delta\binom{q}{2}$.
\end{lemma}

\begin{proof}
	We may assume that the test queries exactly $q$ vertices. For Hamming distance, the statement is well known, and follows easily by taking a union bound over all $\binom{q}{2}$ queried edges. Assume then that $d_e(G,G') \leq \delta$.
	Let $\mu, \mu'$ be the $q$-statistics of $G$, $G'$ respectively, where $G,G' \colon \binom{[n]}{2} \to \Sigma$ are two graphs with earthmover distance at most $\delta$ between them.   
	By Lemma \ref{lem:variation1} it will be enough to show that $|\mu - \mu'| \leq \delta \binom{q}{2}$.  
	Lemma \ref{lem:mixdists} implies that there is an unordered isomorphism $\sigma \colon G \to G'$ with $d_m(\sigma) \leq \delta$.
	
	For any set $Q$ of $q$ vertices, let $\sigma(Q) = \{\sigma(v) : v \in Q\}$, and note that $Q \mapsto \sigma(Q)$ is a bijective mapping from $\mathcal{H}_q$ to itself. Observe that the induced subgraph $G[Q]$ can be non-isomorphic to $G'[\sigma(Q)]$ (as an ordered graph on $q$ vertices) only if there exist two vertices $u,v \in Q$ satisfying $uv \in MS(\sigma)$. By a union bound, the probability of a uniformly random $Q \in \mathcal{H}_q$ to have such a pair is at most $d_m(\sigma) \binom{q}{2} \leq \delta \binom{q}{2}$, implying that $|\mu - \mu'| \leq \delta \binom{q}{2}$.
\end{proof}

The next lemma proves the second (and easier) direction of Theorem \ref{thm:earthmover_tolerant_to_canonical}. It uses Lemma \ref{lem:smallcaninf} to conclude that a canonically testable property is earthmover-resilient and tolerantly testable.
\begin{lemma}
	\label{lem:canonical_to_earthmover}
	Let $\mathcal{P}$ be an ordered graph property. Suppose that $\mathcal{P}$ has a canonical $\epsilon$-test $T$ making $q(\epsilon)$ vertex queries for any $\epsilon > 0$.
	Then $\mathcal{P}$ is $\delta$-earthmover-resilient and $\delta$-tolerantly $\epsilon$-testable with $9q(\epsilon)$ vertex queries, where $\delta(\epsilon) = 1 / 20\binom{q(\epsilon)}{2}$ for any $\epsilon > 0$.
\end{lemma}
\begin{proof}
	Let $\epsilon > 0$, and suppose that $G$ and $G'$ are of earthmover distance at most $\delta(\epsilon)$ between them, where $G$ satisfies $\mathcal{P}$; to prove the earthmover resilience, we need to show that $G'$ is $\epsilon$-close to satisfying $\mathcal{P}$. Since $G \in \mathcal{P}$, it is accepted by $T$ with probability at least $2/3$. By Lemma \ref{lem:smallcaninf}, the acceptance probability of $G'$ by $T$ is at least $\frac{2}{3} - \delta(\epsilon) \binom{q(\epsilon)}{2} > 1/3$. Since $T$ rejects any graph $\epsilon$-far from $\mathcal{P}$ with probability at least $2/3$, we conclude that $G'$ must be $\epsilon$-close to $\mathcal{P}$.
	
	For the second part, regarding tolerant testability, Lemma \ref{lem:smallcaninf} implies that for any graph that is $\delta(\epsilon)$-close to satisfying $\mathcal{P}$, the acceptance probability of $T$ is at least $2/3 - \delta(\epsilon) \binom{q(\epsilon)}{2} > 0.61$. By applying $T$ independently $9$ times and accepting if and only if the majority of the runs accepted, we get a test that accepts $\delta(\epsilon)$-close graphs with probability at least $2/3$ and rejects $\epsilon$-far graphs with probability at least $2/3$ as well. This test can be made canonical with no need for additional queries.
\end{proof}

Let us finish with two comments. First, in the last two lemmas it was implicitly assumed that the canonical test is a deterministic one, but they also hold for randomized ones: The fact that $|\mu - \mu'| \leq \delta \binom{q}{2}$ in Lemma \ref{lem:smallcaninf} is actually enough to imply the statement of Lemma \ref{lem:smallcaninf} for any (deterministic or randomized) canonical test, and Lemma \ref{lem:canonical_to_earthmover} follows accordingly.

Second, the results in this section, along with Sections \ref{sec:earth_tolerant_to_piecewise} and \ref{sec:piecewise_to_canonical}, are not exclusive to two-dimensional structures, and naturally generalize to $k$-dimensional structures for any $k$. Thus, in ordered hypergraphs and tensors in three dimensions or more, it is still true that the combination of earthmover resilience and tolerant testability is equivalent to canonical testability.

\section{Piecewise-canonical testability}
\label{sec:earth_tolerant_to_piecewise}
In this section, we show that ER string properties and ER tolerantly testable ordered graph properties have a constant-query \emph{piecewise canonical test}. 
This is a test that consider a $k$-interval partition of the input, picking a predetermined number of vertices (or entries, in the string case) uniformly at random from each interval (this number may differ between different intervals), and finally, queries all edges between the picked vertices from all intervals.
We always assume that our tolerant tests are non-adaptive and based on $q$ query vertices (we assume they query the entire induced subgraph even if they do not use all of it). Note that unlike the case of unordered graphs, the move from an adaptive test to a non-adaptive one can cause an exponential blowup in the query complexity (we may need to ``unroll'' the entire decision tree).

\begin{definition}
	A (probabilistic) {\em piecewise-canonical} test with $k$ parts and $q$ query vertices for a property $\mathcal{P}$ of functions $f \colon \binom{[n]}{\ell} \to \Sigma$ works as follows. First, the test non-adaptively selects (possibly non-deterministically) numbers $q_1,\ldots,q_k$ that sum up to $q$, and then it considers a $k$-interval partition $I_1, I_2, \ldots, I_k$ of the input function $f$, selecting a uniformly random set of $q_j$ vertices from $I_j$ for every $1\leq j\leq k$. The test finally accepts or rejects $f$ based only on the selected numbers $q_1,\ldots,q_k$ and the unique function $f' \colon \binom{[q]}{\ell} \to \Sigma$ that is isomorphic (in the ordered sense) to the restriction of $f$ on the selected vertices.
	
	A property $\mathcal{P}$ is {\em piecewise-testable} if for for every $\epsilon$ there exist $k(\epsilon)$ and $q(\epsilon)$ for which $\mathcal{P}$ has a piecewise canonical $\epsilon$-test with $k(\epsilon)$ parts and $q(\epsilon)$ query vertices.
\end{definition}

\begin{remark}
	\label{rem:determinstic_canonical_piecewise}
	In Section \ref{sec:prelims} it was noted that a probabilistic canonical test for a property can be transformed into a deterministic one, with the same confidence, as was shown in \cite{GoldreichTrevisan2003}. This is true for any choice of confidence $c$ (not only the ``default'' confidence $c=2/3$). 
	Since one can always amplify a (probabilistic or deterministic) test to get a test of the same type with confidence arbitrarily close to $1$, we conclude that if a property $\mathcal{P}$ has a probabilistic canonical test with a certain confidence $c > 1/2$, then for any $\zeta > 0$, $\mathcal{P}$ has a deterministic canonical test with confidence at least $1-\zeta$.
	
	All of the above is also true for piecewise-canonical tests; the proof for canonical tests carries over naturally to this case, so we omit it. Here, the simulating deterministic test has the same number of parts as the original test.
\end{remark}

\subsection{Strings: Earthmover resilience to piecewise-canonical testability}
\label{subsec:strings_proof1}
In this subsection, we prove that ER properties of strings are piecewise canonically testable. In Section \ref{sec:piecewise_to_canonical}, we show that the latter condition implies canonical testability.

For a string $S:[n] \to \Sigma$ let $d_{\sigma}(S) = |S^{-1}(\sigma)| / n$ denote the \emph{density} of $\sigma$ in $S$. Let $T(S) = (d_{\sigma}(S))_{\sigma \in \Sigma}$ denote the distribution vector of letters in $S$. 
The following well known fact is important for the proof. 
\begin{fact}
	\label{fact:estimate_dist_vector}
	The distribution vector of a string over $\Sigma$ can be approximated up to variation distance $\zeta$, with probability at least $1 - \tau$, using $O(|\Sigma|^2 \log (\tau^{-1})  \zeta^{-2})$ queries. 
\end{fact}

Fix a function $\delta \colon (0,1) \to (0,1)$, a $\delta$-earthmover resilient property $\mathcal{P}$ of strings over $\Sigma$, and $\epsilon > 0$. Take $t = \lceil 1 / 2\delta(\epsilon/2) \rceil$. For any string $S$ over $\Sigma$, let $S_1, \ldots, S_t$ be the $t$-interval partition of $S$ and let the \emph{$t$-interval distribution} $\Gamma_t(S) = (T(S_1), \ldots, T(S_t))$ denote the $t$-tuple of the distribution vectors of $S_1, \ldots, S_t$. For $S$ as above and another string $S'$ over $\Sigma$ with $t$-interval partition $S'_1, \ldots, S'_t$, the \emph{$t$-aggregated distance} between $S$ and $S'$ is $d_A(S, S') = \sum_{i=1}^{t} |T(S_i) - T(S'_i)| \cdot |S_i|/ |S|$; recall that $|T(S_i) - T(S'_i)|$ is the variation distance between $T(S_i)$ and $T(S'_i)$. As usual, we define $d_A(S, \mathcal{P}) = \min_{S' \in \mathcal{P}} d_A(S,S')$.
The next easy lemma relates between the Hamming distance and the $t$-aggregated distance of $S$ to $\mathcal{P}$.
\begin{lemma}
	\label{lem:string_aggregate}
	For any string $S$ over $\Sigma$ we have $0 \leq d_H(S, \mathcal{P}) - d_A(S, \mathcal{P}) \leq \epsilon / 2$.
\end{lemma}
\begin{proof}
	Let $S'$ be the string that is closest to $\mathcal{P}$ among those that can be generated from $S$ only using basic moves inside the intervals $S_1, \ldots, S_t$. In particular, it is trivial that $d_H(S', \mathcal{P}) \leq d_H(S, \mathcal{P})$ and we know by Lemma \ref{lem:mixdists} that $d_e(S, S') \leq 2/t \leq \delta(\epsilon/2)$. By Lemma \ref{lem:earthmover_to_edit}, we get that $d_H(S, \mathcal{P}) - d_H(S', \mathcal{P}) \leq \epsilon/2$. On the other hand, $d_H(S', \mathcal{P}) = d_A(S', \mathcal{P}) = d_A(S, \mathcal{P})$ follows by the definitions of the distance functions and the minimality of $S'$. 
\end{proof}

Finally we present the piecewise canonical test for $\mathcal{P}$. More accurately, we describe a  piecewise-canonical algorithm $\mathcal{A}$ that, given an unknown string $S$ over $\Sigma$ of an unknown length $n$, approximates the $t$-aggregated distance of $S$ to $\mathcal{P}$ up to an additive error of $\epsilon/6$, with probability at least $2/3$. The test simply runs $\mathcal{A}$ and accepts if and only if its output value is at most $\epsilon/4$.
The algorithm $\mathcal{A}$ acts as follows. First, it runs the algorithm of Fact \ref{fact:estimate_dist_vector} in each interval of the $t$-interval partition of $S$, with parameters $\zeta = \epsilon / 6$ and $\tau = 1/3t$. For any $1 \leq i \leq t$, let $T^*_i$ denote the distribution returned by this algorithm for interval $i$. Then, Algorithm $\mathcal{A}$ returns $r = \min_{S' \in \mathcal{P}} \sum_{i=1}^{t} |T^*_i - T_{S'_i}| \cdot |S_i| / |S|$.

With probability $2/3$, we get that $|T(S_i) - T^*_i| \leq \epsilon/6$ for any $i$. Suppose from now on that the latter happens. It follows from the triangle inequality for the variation distance that $d_A(S, \mathcal{P}) \leq d_A(S, S') \leq r + \epsilon/6$, where $r$ is the minimum defined above and $S' \in \mathcal{P}$ is the string achieving this minimum. Conversely, 
there exists $S'' \in \mathcal{P}$ such that $d_A(S, S'') = d_A(S, \mathcal{P})$. But the minimality of $S'$ implies that $\sum_{i=1}^{t} |T^*_i - T_{S''_i}| \cdot |S_i| / |S| \geq r$, and again, from the triangle inequality we get that $d_A(S, S'') \geq r - \epsilon/6$. 
To summarize,
$$
r - \epsilon / 6 \leq d_A(S, S'') = d_A(S, P) \leq d_A(S, S') \leq r + \epsilon/6
$$ 
which means that $r$ is, with probability at least $2/3$, an $(\epsilon/6)$-additive approximation of $d_A(S, \mathcal{P})$. Thus, if $S$ satisfies $\mathcal{P}$ (meaning that $d_A(S, \mathcal{P}) = 0$) then with probability $2/3$ the algorithm $\mathcal{A}$ returns $r \leq \epsilon / 6$ and the test accepts. On the other hand, if $S$ is $\epsilon$-far from $\mathcal{P}$ then $d_A(S, \mathcal{P}) \geq \epsilon / 2$ by the above lemma, and $\mathcal{A}$ returns $r \geq \epsilon / 2 - \epsilon / 6 = \epsilon / 3$ (making the test reject) with probability at least $2/3$, as desired.

\subsection{Ordered graphs: ER and tolerant tests to piecewise-canonical tests}
The next lemma shows that a tolerant test for an ER property $\mathcal{P}$ of ordered graphs can be translated, in an efficient manner, into a piecewise-canonical test for $\mathcal{P}$.

\begin{lemma}\label{lem:topiecewise}
	Let $q:(0,1)\to\mathbb{N}$, $\eta:(0,1)\to(0,1)$, and $\delta:(0,1) \to (0,1)$, and suppose that $\mathcal{P}$ is a $\delta$-earthmover-resilient $\eta$-tolerantly testable property of ordered graphs,
	where for any $\epsilon > 0$, the corresponding $(\epsilon, \eta(\epsilon))$-tolerant test queries $q(\epsilon)$ vertices.
	Then for any $\epsilon > 0$ there exist $q'$ and $k$ such that $\mathcal{P}$ has a piecewise-canonical $\epsilon$-test with $k$ parts and $q'$ query vertices. Moreover, if $q, \eta, \delta$ are polynomial in $\epsilon$, then so are $q'$ and $k$.
\end{lemma}

\begin{proof}
	Let $T$ be a (non-adaptive) $(\epsilon/2, \eta)$-tolerant test for $\mathcal{P}$ querying the induced subgraph on $q' = q(\epsilon/2)$ vertices. Let $G \colon \binom{[n]}{2} \to \Sigma$ denote the unknown input graph.
	Since $T$ is non-adaptive, we may view it as a two-step algorithm acting as follows. In the first step, $T$ chooses a $q'$-tuple $x_1 < \ldots < x_{q'} \in [n]$ (which will eventually be the vertices $T$ will query) according to some distribution $p_T$. The second step receives the tuples $(x_1, \ldots, x_{q'})$ and $(G(x_i x_j))_{i<j \in [q']}$ and decides (probabilistically) whether to accept or reject based only on these tuples.
	
	Take $k = \lceil 2/\delta(\eta(\epsilon/2)) \rceil$ and consider the $k$-interval partition $I_1, \ldots, I_k$ of the input graph $G$. 
	Our piecewise-canonical test $T'$, also making $q'$ vertex queries, is designed as follows. 
	First it picks a tuple $X$ of $q'$ elements $x_1 < \ldots < x_{q'} \in [n]$ according to the distribution $p_T$. For each $i = 1,\ldots,k$, let $q_i = |X \cap I_i|$ and let $S_i = \{ 1 + \sum_{j=1}^{i-1} q_j, \ldots, \sum_{j=1}^{i} q_j\}$. $T'$ queries exactly $q_i$ vertices from $I_i$ uniformly at random. Now, $T'$ picks a permutation $\pi \colon [q'] \to [q']$ in the following manner: For each $1 \leq i \leq k$, $\pi$ restricted to $S_i$ is a uniformly random permutation on $[S_i]$.
	Finally, $T'$ runs the second step of the original test $T$, with tuples $(x_1, \ldots, x_{q'})$ and $(G(x_{\pi(i)} x_{\pi(j)}))_{i<j \in [q']}$.
	
	Clearly, $T'$ makes in total $q'$ queries in $k$ intervals, where the vertex queries within each interval are chosen uniformly at random. It only remains to show that $T'$ is a valid $\epsilon$-test. 
	Observe that applying $T'$ on the input graph $G$ is equivalent to the following process, in the sense that their output distribution (given any fixed $G$) is identical.
	\begin{enumerate}
		\item ``Shuffle'' the vertices inside each interval $I_i$ of $G$ in a uniformly random manner, to get a new ordered graph $G'$.
		\item Run the original test $T$ on $G'$, and return its answer.
	\end{enumerate}
	The relative mixingness between $G$ and any such $G'$ is at most $k \binom{\lceil n/k \rceil}{2} / \binom{n}{2} < 2 / k \leq \delta(\eta(\epsilon/2))$ where the first inequality holds for large enough $n$. By Lemmas \ref{lem:mixdists} and \ref{lem:earthmover_to_edit} and the $\delta$-earthmover resilience of $\mathcal{P}$, we get that $|d_H(G', \mathcal{P}) - d_H(G, \mathcal{P})| \leq \eta(\epsilon / 2) < \epsilon/2$. Thus, if $G$ satisfies $\mathcal{P}$, then any $G'$ possibly generated in the first step of the above process is $\eta(\epsilon / 2)$-close to $\mathcal{P}$.
	Since $T$ is $(\epsilon/2, \eta)$-tolerant, the second step of the process accepts with probability at least $2/3$ for any fixed choice of $G'$. Thus, the process (or equivalently, $T'$) accepts $G$ with probability at least $2/3$ in this case.
	Conversely, if $G$ is $\epsilon$-far from $\mathcal{P}$ then $G'$ generated in the first step is $\epsilon/2$-far from $\mathcal{P}$, and, similarly, the process (or equivalently, $T'$) rejects with probability at least $2/3$.
\end{proof}

\section{Piecewise-canonical testability to canonical testability}
\label{sec:piecewise_to_canonical}
This section is dedicated to the proof that piecewise-canonically testable properties are canonically testable. While the proofs are presented here for ordered graphs, they can easily be translated to the case of strings. Therefore, the results in this section, combined with the previous two sections, complete the proof of Theorems \ref{thm:earthmover_tolerant_to_canonical} and \ref{thm:string_characterization}.

\begin{definition}
	Given $\{q_1,\ldots,q_k\}$ that sum up to $q$ and $t\geq\max_{1\leq j\leq k}q_j$, the {\em $t$-simulated piecewise distribution} over subsets of of $[n]$ of size $q$ is the result of the following process.
	\begin{description}
		\item[Uniform sampling] Select a set of $tk$ indices from $[n]$, uniformly at random. Let $\{i_1,\ldots,i_{tk}\}$ denote the set with its members sorted in ascending order.
		\item[Simulation inside each block] For every $1\leq j\leq k$, select a subset of $\{i_{(j-1)t+1},\ldots,i_{jt}\}$ of size $q_i$, uniformly at random.
	\end{description}
\end{definition}

\begin{lemma}\label{lem:simupiece}
	For every $\delta$, $k$ and $q$, there exist $t(\delta,k,q)$ and $N(\delta,k,q)$ polynomial in $\delta, k, q$, so that if $n>N(\delta,k,q)$ then the $t$-simulated piecewise distribution with respect to $q_1,\ldots,q_k$ is $\delta$-close (in the variation distance) to an actual piecewise distribution with respect to $q_1, \ldots, q_k$, i.e., a process of the following type. 
	Consider a $k$-interval partition $I_1, \ldots, I_k$ of the input graph, and for every $1\leq j\leq k$, pick a uniformly random subset of $I_j$ of size $q_j$.
\end{lemma}
In the proof of Lemma \ref{lem:simupiece} we do not try to optimize the dependence of $t$ and $N$ on $\delta, k, q$, but just show that it is a reasonable polynomial dependence. 
\begin{proof}
	Fix $q_1, \ldots, q_k$ and write $Q_i = \sum_{j=1}^{i} q_j$ for any $1 \leq i \leq k$. Also take $q = Q_k$. 
	For any $1 \leq l_1 < \ldots < l_q \leq n$	
	denote by $\Pr_{\text{piece}}(E_{l_1, \ldots, l_q})$ the probability that the indices selected by a piecewise canonical distribution with parameters $q_1, \ldots, q_k$ are $l_1, \ldots, l_q$.
	Similarly, for $q_1, \ldots, q_k$ as above and a fixed $t \geq \max_{1\leq j\leq k} q_j$, we denote by $\Pr_{\text{sim}}(E_{l_1, \ldots, l_q})$ the probability that the indices selected by a simulated piecewise canonical distribution with parameters $q_1, \ldots, q_k$ and $t$ are $l_1, \ldots, l_q$.
	It is enough to show, for a suitable choice of $t$ and for $n$ large enough, that $\sum_{l_1 < \ldots < l_q} |\Pr_{\text{piece}}(E_{l_1, \ldots, l_q}) - \Pr_{\text{sim}}(E_{l_1, \ldots, l_q}) | < \delta$.
	To prove this, we show that there exist suitable events $A$ and $B$ satisfying the following conditions.
	\begin{itemize}
		\item $\Pr_{\text{piece}}(A) \leq \delta$ and $\Pr_{\text{sim}}(B) \leq \delta$.
		\item $\Pr_{\text{piece}}(E_{l_1, \ldots, l_q} | \neg A) = \Pr_{\text{sim}}(E_{l_1, \ldots, l_q} | \neg B)$ for any possible choice of $l_1 < \ldots < l_q$, where $\neg A$ and $\neg B$ are the complementary events of $A$ and $B$, respectively.
	\end{itemize}  
	
	In the rest of the proof we define and analyze the events $A$ and $B$.
	\subparagraph{Order statistics}
	Take $t = 600 k^4 q^2 \delta^{-3}$ and $N = tk$. Let $1 \leq i_1 < \ldots < i_N \leq n$ be the elements of an $N$-tuple from $\binom{[n]}{N}$, picked uniformly at random. It is well known (see, e.g., Chapter 3 in \cite{ArnoldBalakNagar1992}) that the expected value of $i_r$ -- the $r$-th order statistic of the tuple -- is $\mu_r = r(n+1)/(N+1)$ and satisfies $|\mu_r - rn/N| < n/N$, and the variance of $i_r$ is $\sigma_r^2 \leq n^2 / N$.
	
	By Chebyshev's inequality, for any $1 \leq r \leq N$ it holds that $\Pr(|i_r - \mu_r| > \alpha n) < 1 / N \alpha^2 $. Pick $\alpha = 3\sqrt{k / \delta N} < \delta / 8 k^2 q$.
	For any $1 \leq j \leq k-1$,
	we take $r_{j}^{-}$ as the largest integer $r$ for which $\mu_r < (Q_j / k - \alpha - 1/N)n$ and $r_{j}^{+}$ as the smallest integer $r'$ for which $\mu_{r'} > (Q_j / k + \alpha + 1/N)n$; note that $tQ_j - r_j^- < 2 \alpha N$ and $\mu_{r_{j}^{-}} > (Q_j / k - 2 \alpha) n$, and on the other hand, $r_j^+ - tQ_j < 2 \alpha N$ and $\mu_{r_{j}^{+}} < (Q_j / k + 2 \alpha) n$. Intuitively speaking, $r_j^-, r_j^+$ were chosen here with the following requirements in mind. With good probability, $r_j^-$ needs to be contained in $I_j$, $r_j^+$ needs to be contained in $I_{j+1}$, and both $r_j^-$ and $r_j^+$ should be close to $jn/k$ (which is roughly equal to the last element of $I_j$ and the first element of $I_{j+1}$).
	  
	Indeed, let $C$ denote the event that
	\begin{align}
	\left(\frac{Q_j}{k} - 3\alpha \right)n < i_{r_{j}^{-}} < \frac{Q_j}{k} n < i_{r_{j}^{+}} < \left(\frac{Q_j}{k} + 3 \alpha \right)n
	\end{align} 
	holds for any $1 \leq j \leq k-1$, and observe that $\left(\frac{Q_j}{k} + 3 \alpha \right)n < \left(\frac{Q_{j+1}}{k} - 3 \alpha \right)n$ for any $j$. $\neg C$ is contained in the event that, for some $j$, $|i_{r_{j}^{-}} - \mu_{r_{j}^{-}}| > \alpha n$ or $|i_{r_{j}^{+}} - \mu_{r_{j}^{+}}| > \alpha n$. The probability of the latter event is bounded by $2k/N\alpha^2 = 2\delta/9$ by a union bound. Therefore $C$ holds with probability at least $1 - 2\delta/9$.

	\subparagraph{The ``bad'' events $A$ and $B$}
	Suppose that, after picking $i_1 < \ldots < i_N$ uniformly at random as above, we pick two (not necessarily disjoint) $q$-tuples $w ,w'$ of vertices from $[n]$ simultaneously: $w$ is picked according to the piecewise canonical distribution among all elements of $G$, whereas $w'$ is picked according to the $t$-simulated piecewise distribution, considering $\{i_1, \ldots, i_N\}$ as the output of the first step -- the \emph{uniform sampling} step -- of the simulated process.
	The events $A$ and $B$ are defined as follows. $A$ holds if and only if either $C$ doesn't hold  or some entry of $w$ is picked from $I = \bigcup_{i=1}^{k-1} I_j$, where $I_j = \{i_{r_j^-}, i_{r_j^{-}} + 1, \ldots, i_{r_j^+}\}$ for any $j$. $B$ holds if and only if either $C$ doesn't hold or some entry of $w'$ is taken from $I' = \bigcup_{j=1}^{k-1} I'_j$, where $I'_j =  \{i_{r_j^{-}}, i_{r_j^{-} + 1}, \ldots, i_{r_j^{+}}\}$ for any $j$.
	
	\subparagraph{$A$ and $B$ satisfy the requirements}
	The major observation here is that the distribution of the piecewise canonical
	distribution under the assumption that $A$ does not hold is identical to the distribution of the simulated process under the assumption that $B$ does not hold. 
	That is, $\Pr_{\text{piece}}(E_{l_1, \ldots, l_q} | \neg A) = \Pr_{\text{sim}}(E_{l_1, \ldots, l_q} | \neg B)$ for any possible choice of $l_1 < \ldots < l_q$, as required above.
	To see this, observe that under these assumptions, both distributions pick exactly $q_j$ entries, uniformly at random, from the set $\{i_{r_j^+}+1, \ldots, i_{r_{j+1}^-}-1\}$ for any $0 \leq j \leq k-1$ (where we define $r_0^+ = 0$ and $r_k^- = n+1$).  
	It remains to show that $\Pr_{\text{piece}}(A) \leq \delta$ and $\Pr_{\text{sim}}(B) \leq \delta$. 
	
	In the piecewise-canonical distribution, every entry has probability at most $q/n$ to be picked. Assuming that $C$ holds, we get that $|I_j| < 6 \alpha n$ for any $j$, and so $|I| \leq 6 \alpha k n$. Therefore, $\Pr_{\text{piece}}(A | C) \leq |I| q / n < 6 \alpha k q < 3 \delta / 4$. Thus, $\Pr_{\text{piece}}(A) \leq \Pr_{\text{piece}}(A | C) + \Pr(\neg C) < \delta$, as needed.
	
	In the simulated distribution, the probability that any given element from $I'$ is taken to $w'$ is at most $q/t$. Since $|I_j'| < 4 \alpha N$, we get that $|I'| < 4 \alpha k N$ and so $\Pr_{\text{sim}}(B) \leq |I'| q / t + \Pr(C) < 4 \alpha k^2 q + \delta/4 < \delta$, as desired.  
\end{proof}
\begin{lemma}\label{lem:tocanonical}
	A piecewise-testable property has a canonical test. Moreover, if the number of parts of the piecewise-canonical $\epsilon$-test, denoted by $k(\epsilon)$, and its number of vertex queries, denoted by $q(\epsilon)$, are polynomial in $\epsilon$, then so is the number of queries of the canonical test.
\end{lemma}

\begin{proof}
	Let $\mathcal{P}$ be a piecewise-testable property.
	Following Remark \ref{rem:determinstic_canonical_piecewise}, for any $\epsilon > 0$ there exists a deterministic piecewise-canonical $\epsilon$-test $T$, with confidence $3/4$, making exactly $q$ queries on $k$ parts. 
	To simulate $T$ using a canonical test $T'$, we pick $\delta = 1/12$ and take $t = t(\delta, k, q)$ as provided by Lemma \ref{lem:simupiece} (here we also implicitly assume that $n > N(\delta, k, q)$). $T'$ is taken as the $t$-simulated piecewise test, that queries the induced subgraph $H$ on $kt$ vertices picked uniformly at random, and then imitates $T$: If, for any $1 \leq i \leq k$, $T$ chooses $q_i$ vertices in part number $i$, then $T'$ chooses $q$ vertices of $H$ using a $t$-simulated piecewise distribution, where $q_i$ vertices are taken from the $i$-th simulated block. Then, $T'$ makes the same decision that $T$ would have made on the queried subgraph induced on the chosen $q$ vertices.
	
	By Lemma \ref{lem:simupiece}, the distributions $\eta$ and $\eta'$ over $q$-tuples of vertices generated by $T$ and $T'$, respectively, are $\delta$-close.
	Let $\mathcal{H}$ be a family of ordered graphs on $q$ vertices such that $T$ accepts its queried induced subgraph $H$ if and only if $H \in \mathcal{H}$. Then, $\Pr_{\eta}(H \in \mathcal{H}) \geq 3/4$ if the input graph $G$ satisfies $\mathcal{P}$, whereas $\Pr_{\eta}(H \in \mathcal{H}) < 1/4$ if $G$ is $\epsilon$-far from $\mathcal{P}$. By Lemma \ref{lem:variation1}, if the input graph $G$ for $T'$ satisfies $\mathcal{P}$ then the queried induced subgraph $H$ satisfies $\Pr_{\eta'}(H \in \mathcal{H}) \geq 3/4 - \delta = 2/3$, and if $G$
	is $\epsilon$-far from $\mathcal{P}$, then $\Pr_{\eta'}(H \in \mathcal{H}) < 1/4 + \delta = 1/3$. Thus, $T'$ is a valid test for $\mathcal{P}$.
\end{proof}

Lemmas \ref{lem:topiecewise} and \ref{lem:tocanonical} together prove the first (and more difficult) direction of Theorem \ref{thm:earthmover_tolerant_to_canonical}.

\section{Canonical testability to estimability}
\label{sec:canonical_to_estimable}
This section describes the proof of Theorem \ref{thm:canonical_to_estimable}. The proof takes roughly the same steps as in the proof of Fischer and Newman \cite{FischerNewman2005} for the unordered case.
For the proof of \cite{FischerNewman2005} to work in our case, we only need to make a few slight modifications. Therefore, instead of rewriting the whole proof, we only describe what modifications are made and how they change the proof. 

The proof in \cite{FischerNewman2005} builds on a test for partition parameters, established in the seminal paper of Goldreich, Goldwasser and Ron \cite{GoldreichGoldwasserRon1998}. The test of \cite{GoldreichGoldwasserRon1998} also needs to be slightly modified for our needs. Therefore, the partition test receives the same treatment as the proof in \cite{FischerNewman2005}: We describe the modified statement and how to change the proof accordingly, but do not get into unnecessary technicalities.  

\subsection{The unordered proof}
First we sketch the proof that canonical testability in unordered graphs implies estimability \cite{FischerNewman2005}.
\subsubsection{Signatures of regular partitions and approximating the $q$-statistic}
\label{subsubsec:signatures}
A $(\gamma, \epsilon)$-signature for an equipartition $\mathcal{A} = \{ V_1, \ldots, V_t \}$ is a sequence of densities $\eta_{i,j}$, such that the density between $V_i$ and $V_j$ differs from $\eta_{i,j}$ by at most $\gamma$, for all but at most $\epsilon \binom{t}{2}$ of the pairs $i,j$.
The (labeled) $q$-statistic of a graph is the distribution of the labeled graphs on $q$ vertices in it. Given a signature as above, it is natural to define the perceived $q$-statistic of the signature as the distribution on labeled $q$-vertex graphs generated as follows: First we choose $q$ indices $i_1, \ldots, i_q$ from $[t]$. Then for every $j < j'$ we add an edge between $v_j$ and $v_{j'}$ with probability $\eta_{i_j, i_{j'}}$, independently. 
The main observation in this part is that the perceived $q$-statistic of a signature with good (small) enough parameters of a regular enough partition of a graph $G$ is close to the actual $q$-statistic of $G$. Thus, to estimate the $q$-statistic of a graph we just need to obtain a good signature of a regular partition of this graph. For more details, see Section 4 in \cite{FischerNewman2005}.

\subsubsection{Computing signature of a final partition}
Implicit in the proof of the celebrated Szemer\'edi regularity lemma \cite{Szemeredi1976} is the concept of an \emph{index} of an equipartition, which is a convex function of partitions that never decreases under taking refinements of a partition. A partition $P$ is \emph{robust} if, for any refinement $Q$ of $P$ that is not too large (in terms of the number of parts) with respect to $P$, the index of $Q$ is similar to that of $P$.
The main argument in \cite{Szemeredi1976} is that robustness implies regularity. An even stronger condition, that implies robustness, is finality. A partition $P$ is \emph{final} if for \emph{any} partition $Q$\footnote{Here $Q$ is not necessarily a refinement of $P$} whose number of parts is not much larger than that of $P$, the index of $Q$ is also not much larger than that of $P$. It is easy to prove that robust and final partitions with arbitrarily good parameters exist. The definitions appear in Section 4 of \cite{FischerNewman2005}, while the rest of the discussion here appears in Section 5 there.

Knowing the parameters of a good signature of a robust enough partition is useful for estimation, as we shall see soon. Before doing so, we explain how to find such a signature using the partition parameters test of \cite{GoldreichGoldwasserRon1998}. This test is described in a more formal and detailed fashion in Subsection \ref{subsec:partition_parameters}, but for our purposes, it acts as a test for the property of ``having a given signature''. We consider a quantized set of signatures, which contains only a constant number of possible signatures, so that every graph is close to a graph satisfying one of the signatures (i.e., an $\eta$-net for a suitable parameter $\eta$).

By applying the test of \cite{GoldreichGoldwasserRon1998} to each of the signatures sufficiently many times and accepting or rejecting each of the signatures according to majority vote, we determine with good probability which signatures our input graph $G$ is close to having. More precisely, all signatures that are very close to some actual signature $S$ of $G$ are accepted, and all of those that are very far from any actual signature $S$ are rejected. Thus, this process only accepts signatures that are at the very least ``quite close'' to some actual one.

Finally, an index measure can also be defined for signatures, and the index of a good signature is close to that of the corresponding partition. Under the assumption that all signatures that we captured are quite close to an actual one, in particular we will find a good approximation of a final partition, and will recognize that it is final by not finding signatures of partitions that are only somewhat bigger and have a much bigger index (meaning that such partitions do not exist).

\subsubsection{Knowing signature of a robust partition implies estimation}
Note that for $\delta > 0$ and a family $\mathcal{H}$ of $q$-vertex graphs, having only a good signature $S$ of a robust enough partition allows us to distinguish for any $\epsilon > 0$, deterministically, between the case that $G$ is $(\epsilon - \delta)$-close to a graph $G'$ that contains a large number of copies of labeled graphs from $\mathcal{H}$, and the case that all graphs that are $\epsilon$-close to $G$ contain only a small number of $\mathcal{H}$-copies. 
Combining this statement with the one from the previous subsection, stating that computing the signature of some robust (and in particular, final) partition is possible with good probability in constant time, it is straightforward to conclude that any testable graph property is estimable. 
As the proof of this statement is rather technical and the main arguments do not change when moving to the ordered case, we do not go into the details of the proof here. Section 6 in \cite{FischerNewman2005} is dedicated to this proof.
\subsection{Adapting to the ordered setting}
\label{subsec:adapt_canonical_estimable}
Suppose that a property $\mathcal{P}$ has a canonical test making $q$ queries. Using the proof for the undirected case as is will not work here. The reason is that, theoretically, a pair of vertex sets can be regular as an unordered pair, but interleaved in a way that makes it useless when we are interested in understanding the ordered $q$-statistic of a graph. Another issue that needs to be considered is the fact that we work here with edge-colored graphs, instead of standard ones. However, the latter is not a real issue: As observed in previous works \cite{AlonBenEliezerFischer2017, AustinTao2010, AxenovichMartin2011}, regularity-based arguments tend to generalize in a straightforward manner to the multicolored setting.

To accommodate for the first issue, we need a ``regularity scheme'' that is slightly different from the unordered instance. At the base of the scheme lies a $k$-interval equipartition $I$ for a suitable $k$, which is known in advance. The regular, robust or final partitions that we need along the proof (analogously to the unordered case) are always refinements of the interval equipartition $I$, where we do not care about the relation between two parts that lie inside the same interval. Here, for partitions $P$ and $Q$, we say that $Q$ is a \emph{refinement} of $P$ if any part of $Q$ is completely contained in a part of $P$.
A formal presentation of the scheme is given in the next few definitions and lemmas.
The first definition presents the $(q,k)$-statistic of a graph, which in some sense is the $k$-partite version of the $q$-statistic, as defined in Section \ref{sec:earthmover_mixing}.
\begin{definition}
Let $G, H$ be $\Sigma$-edge-colored ordered graphs on $n \geq q$ vertices respectively, and let $I = I_k(G) = (I_1, \ldots, I_k)$ be the $k$-interval equipartition of $G$ for $k \geq q$.
A $q$-vertex induced subgraph of $G$ is $k$-separated if, for every $1 \leq i \leq k$, no two vertices of the subgraph lie in $I_i$. The total number of $k$-separated subgraphs on $q$ vertices in a graph on $n$ vertices is denoted by $N(k,q,n)$. 
The number of $k$-separated $H$-copies in $G$ is denoted by $h_k(H, G)$. The \emph{$k$-density} of $H$ in $G$ is $t_k(H,G) = h_k(H, G) / N(k, q, n)$.
Finally, the \emph{$(q,k)$-statistic} of $G$ is the vector $(t_k(H, G))_{H \in \mathcal{H}_q}$, where $\mathcal{H}_q$ is the family of all $\Sigma$-edge-colored ordered graphs with $q$ vertices.
\end{definition}

\begin{observation}
\label{obs:q_stat_vs_q_k_stat}
The variation distance between the $q$-statistic and the $(q,k)$-statistic of a graph is at most $q^2/2k$.
\end{observation}
\begin{proof}
For a uniformly chosen pair $(u,v)$ of disjoint vertices in a graph $G$, the probability that $v$ lies in the same interval as $u$ is at most $\frac{n/k}{n-1}$. By a union bound, the probability that a uniformly random $q$-tuple $Q$ of disjoint vertices contains two vertices in the same interval is at most 
$$
\frac{n}{k(n-1)} \binom{q}{2} \leq \frac{q}{k(q-1)} \binom{q}{2} = \frac{q^2}{2k}.
$$
Conditioning on the above not happening, the induced subgraph $G[Q]$ is distributed according to the $(q,k)$-statistic. The statement of the lemma thus follows from Lemma \ref{lem:variation1}.
\end{proof}
The next definition presents the $k$-partite notion analogous to canonical testability. 
\begin{definition}
	\label{def:k_canonical}
A property $\mathcal{P}$ of $\Sigma$-edge-colored ordered graphs is \emph{$(\epsilon, q, k)$-canonical} if there exists a set $\mathcal{A}$ of $q$-vertex $\Sigma$-edge-colored ordered graphs satisfying the following two conditions.
\begin{itemize}
	\item If an ordered graph $G$ satisfies $\mathcal{P}$, then $\sum_{H \in \mathcal{A}} t_k(H, G) \geq 2/3$. In this case we say that $G$ is $\mathcal{A}$-\emph{positive}.
	\item If $G$ is $\epsilon$-far from satisfying $\mathcal{P}$, then $\sum_{H \in \mathcal{A}} t_k(H, G) \leq 1/3$. Here $G$ is $\mathcal{A}-$\emph{negative}.
\end{itemize}
\end{definition}
Note that there may be graphs that are neither positive not negative with respect to $\mathcal{A}$ in the above definition.
As it turns out, canonical $\epsilon$-testability implies $(\epsilon, q, k)$-canonicality for a  suitable $q$ and any $k = \Omega(q^2)$. In fact, the converse is also true, but is not needed for our proof.
\begin{lemma}
\label{lem:canonical_to_interval_canonical}
If a property $\mathcal{P}$ of edge-colored ordered graphs is canonically testable, then there exists a function $q:(0,1) \to \mathbb{N}$ so that $\mathcal{P}$ is $(\epsilon, q(\epsilon), k)$-canonical for any $\epsilon > 0$ and $k \geq 4q(\epsilon)^2$.
\end{lemma}
\begin{proof}
By Remark \ref{rem:determinstic_canonical_piecewise}, if $\mathcal{P}$ is canonically testable then for any $\epsilon > 0$ it has a canonical $\epsilon$-test with confidence $11/12$, making $q = q(\epsilon)$ queries. This means that there is a family $\mathcal{A}$ of $q$-vertex graphs, such that $\sum_{H \in \mathcal{A}} t(H, G) \geq 11/12$ for graphs $G$ satisfying $\mathcal{P}$ and $\sum_{H \in \mathcal{A}} t(H, G) \leq 1/12$ for graphs that are $\epsilon$-far from $\mathcal{P}$. By Observation \ref{obs:q_stat_vs_q_k_stat}, $\sum_{H \in \mathcal{A}} |t(H, G) - t_k(H, G)| \leq 2 \frac{q^2}{2k} \leq 1/4$, and the statement follows.
\end{proof}

The definition of a regular partition needed for our case is given below. Here, the partition must refine the base interval equipartition, and we do not care how parts inside the same interval interact between themselves. For a single pair of parts lying in different intervals, the notion of regularity that we use is the standard multicolored notion, defined in Subsection \ref{subsec:regular_reducible}.
\begin{definition}[$k$-refinement, $(\gamma, k)$-regular partition, $(\gamma, \epsilon, k)$-signature]
Let $G$ be an $\Sigma$-edge-colored ordered graph, and let $I = (I_1, \ldots, I_k)$ be the $k$-interval equipartition of $G$.
An equipartition $P = (V_{11}, \ldots, V_{1r}, \ldots, V_{k1}, \ldots, V_{kr})$ is a $k$-refinement if $V_{ij} \subseteq I_i$ for any $i,j$. $P$ is \emph{$(\gamma, k)$-regular} if it is a $k$-refinement and all but a $\gamma$-fraction of the pairs $(V_{ij}, V_{i'j'})$ with $i < i'$ are $\gamma$-regular.

A \emph{$(\gamma, \epsilon, k)$-signature} of $P$ is a sequence $S = (\eta_{ij}^{i'j'}(\sigma))$ for $i < i' \in [k]$, $j,j' \in [r]$, $\sigma \in \Sigma$, such that for all but at an $\epsilon$-fraction of the pairs $(V_{ij}, V_{i'j'})$ with $i < i'$, we have $|d_{\sigma}(V_{ij}, V_{i'j'}) - \eta_{ij}^{i'j'}(\sigma)| \leq \gamma$ for any $\sigma \in \Sigma$.
A $(\gamma, \gamma, k)$-signature is also referred to as a $(\gamma, k)$-signature.
\end{definition}
In the above definition, $d_{\sigma}(U,V)$ is the density of the color $\sigma$ among edges between $U$ and $V$.
The \emph{perceived $(q, k)$-statistic} is the natural translation of the notion of the perceived $q$-statistic from Definition 7 in \cite{FischerNewman2005} to our $k$-partite setting: It captures the ``expected'' fractions of each of the graphs on $q$ vertices among the $k$-separated $q$-vertex subgraphs of $G$.
$(f, \gamma, k)$-Robust and $(f, \gamma, k)$-final partitions (see Section 4 in \cite{FischerNewman2005} for the original unordered definitions) are also defined with respect to the $k$-partite structure, where we do not care about the relation between pairs of parts from the same interval. To accommodate the fact that we consider multicolored graphs, the \emph{index} of a pair $U,V$ is $\sum_{\sigma \in \Sigma} d_{\sigma}(U,V)^2$ (compared to $d(U,V)^2$ in the case of standard graphs). The index of an equipartition refining an interval partition is the sum of indices of all pairs not coming from the same interval, divided by the total number of such pairs.

After providing the definitions required for our ordered setting, the main statements of the proof, analogous to Lemmas 3.8, 4.4 and 4.5 in \cite{FischerNewman2005}, are the following.
\begin{lemma}[Ordered analogue of Lemma 3.8 in \cite{FischerNewman2005}]
For every $q$ and $\epsilon$ there exist $\gamma$ and $k$, so that for every $(\gamma, k)$-regular partition $P$ of $G$ into $t \geq k$ sets, where $G$ has $n \geq N(q, \epsilon, t)$ vertices, and for every $(\gamma, k)$-signature $S$ of $P$, the variation distance between the actual $(q,k)$-statistic and the perceived $(q,k)$-statistic with respect to $S$ is at most $\epsilon$. 
\end{lemma}

\begin{lemma}[Ordered analogue of Lemma 4.4 in \cite{FischerNewman2005}]
\label{lem:finding_robust_paramteres}
For every $k$, $\gamma$, and $f \colon \mathbb{N} \to \mathbb{N}$ there exist $q$, $T$, and an algorithm that makes up to $q$ (piecewise-canonical) queries to any large enough graph $G$, computing with probability at least $2/3$ a $(\gamma, k)$-signature of an $(f, \gamma, k)$-final partition of $G$ into at most $T$ sets.
\end{lemma}
Note that the second lemma requires piecewise-canonical vertex queries, making our algorithm a piecewise-canonical one. But Lemma \ref{lem:tocanonical} implies that this algorithm can be converted into a canonical one, since an algorithm that distinguishes between $\delta$-closeness to a property $\mathcal{P}$ and $\epsilon$-farness from $\mathcal{P}$, for any $\epsilon > \delta$, is actually an $(\epsilon - \delta)$-test for being $\delta$-close to $\mathcal{P}$.

\begin{lemma}[Ordered analogue of Lemma 4.5 in \cite{FischerNewman2005}]
\label{lem:robust_is_enough}
	For every $q$ and $\delta$ there exist $\gamma$, $k$, and $f \colon \mathbb{N} \to \mathbb{N}$ with the following property. For every family $\mathcal{H}$ of edge-colored ordered graphs with $q$ vertices there exists a deterministic algorithm that receives as an input only a $(\gamma, k)$-signature $S$ of an $(f, \gamma, k)$-robust partition with $t \geq k$ sets of a graph $G$ with $n \geq N(q, \delta, t)$ vertices, and distinguishes given any $\epsilon$ between the case that $G$ is $(\epsilon - \delta)$ close to some $\mathcal{H}$-positive graph, and the case that $G$ is $\epsilon$-far from every graph that is not $\mathcal{H}$-negative.
\end{lemma}
Once all definitions for our setting have been given, Lemma $\ref{lem:canonical_to_interval_canonical}$ brings us to a ``starting point'' from which the flow of the proof is essentially the same as in the unordered case, other then two issues mentioned and handled below. To avoid repeating the same ideas as in the unordered case, we will not provide the full technical details of the proofs of the three main lemmas. Deriving the proof of Theorem \ref{thm:canonical_to_estimable} from Lemmas \ref{lem:finding_robust_paramteres} and \ref{lem:robust_is_enough} is similar to the unordered case.

One place where the move to a multicolored version requires more care is in proving the multicolored analogue of Lemma 6.2 in \cite{FischerNewman2005}. In the original proof, edges are being added/removed with a suitable probability, where the decision whether to modify an edge is independent of the other edges. In the multicolored version, the analogue of adding/removing edges is recoloring them. One way to do this is the following: for every color $c$ where edges need to be added, we consider every relevant edge that has a ``too dense'' color $c'$ and, with a suitable probability (that depends on the densities of the colors $c,c'$ and the relevant signature), we recolor this edge from $c'$ to $c$. By doing this iteratively for all colors that are in deficit, the multicolored analogue of Lemma 6.2 in \cite{FischerNewman2005} follows.

Another issue is that for our ordered setting, we need a ``partition parameters'' test that is slightly different than the one proved in \cite{GoldreichGoldwasserRon1998} and used in \cite{FischerNewman2005}. We describe the modified partition parameters problem in Subsection \ref{subsec:partition_parameters}.

\subsection{The partition parameters test}
\label{subsec:partition_parameters}
Let $\Phi = \{\rho_j^{LB}, \rho_j^{UB}\}_{j=1}^{k} \cup \{ \varrho_{j,j'}^{LB}, \varrho_{j,j'}^{UB}\}_{j < j' \in \binom{[k]}{2}}$ be a set of nonnegative parameters so that $\rho_j^{LB} \leq \rho_j^{UB}$ and $\varrho_{j,j'}^{LB} \leq \varrho_{j,j'}^{UB}$. An $n$-vertex graph $G = (V,E)$ satisfies an (unordered) \emph{$\Phi$-instance} if there is a partition $V = V_1 \cup \ldots \cup V_k \cup V'$ such that 
\begin{itemize}
	\item $0 \leq |V'| < k$ and $|V|-|V'|$ is divisible by $k$.
	\item For any $1 \leq j \leq k$, $\rho_j^{LB} \lfloor n / k \rfloor \leq |V_j| \leq  \rho_j^{UB} \lfloor n / k \rfloor$.
	\item For any $j < j' \in \binom{[k]}{2}$, $\varrho_{j,j'}^{LB} \lfloor  n / k \rfloor^2 \leq |E[V_i, V_j]|  \leq \varrho_{j,j'}^{UB}  \lfloor n / k \rfloor^2$. 
	\end{itemize}
In \cite{GoldreichGoldwasserRon1998}, it was shown that the property of having an unordered $\Phi$-instance is testable. 

For our purposes, the base graph that we need to consider is an edge-colored $r$-partite graph, where the parts are of equal size (instead of a complete base graph, as in the unordered case). Formally, the partition parameters problem that we need to test is the following. 
\begin{definition}[Ordered $\Phi$-instance]
An \emph{ordered $\Phi$-instance} whose parameters are the positive integers $r$ and $k$ and the finite color set $\Sigma$ consists of the following ingredients:
\begin{itemize}
\item For every $i < i' \in [r]$ and $j,j' \in [k]$ and every $\sigma \in \Sigma$, there are parameters $ \ell_{j,j'}^{i,i'}(\sigma) \leq  h_{j,j'}^{i,i'}(\sigma)$.
\item For fixed $i,i',j,j'$, it holds that $\sum_{\sigma \in \Sigma} \ell_{j,j'}^{i,i'}(\sigma) \leq 1 \leq \sum_{\sigma \in \Sigma} h_{j,j'}^{i,i'}(\sigma)$.
\end{itemize}
Let $G$ be an $n$-vertex $\Sigma$-edge-colored ordered graph, and denote its $r$-interval equipartition by $I = (I_1, \ldots, I_r)$. 
$G$ is said to \emph{satisfy} $\Phi$ if there exist disjoint sets of vertices $V_{11}, \ldots, V_{1k}, \ldots, V_{r1}, \ldots, V_{rk}$ such that for any $i$ and $j$, $V_{ij} \subseteq I_i$ and $|V_{ij}| = \lfloor n / rk \rfloor$, and $\ell_{j,j'}^{i,i'}(\sigma) \leq d_{\sigma}(V_{ij}) \leq  h_{j,j'}^{i,i'}(\sigma)$ for any $i < i' \in [r]$, $j,j' \in [k]$ and $\sigma \in \Sigma$.
\end{definition}
Recall that $d_{\sigma}(A,B)$ is the \emph{density} function of the color $\sigma$ between the sets $A$ and $B$.
Note that while in the original unordered $\Phi$-instance, one could also specify lower and upper bounds on the number of vertices in each part, in our case it is not needed; for us it suffices to consider the special case where the size of each part is a $1/rk$-fraction of the total number of vertices.
\begin{lemma}
The edge-colored ordered graph property of satisfying an ordered $\Phi$-instance is testable.
\end{lemma}
The proof is very similar to that of the unordered case in \cite{GoldreichGoldwasserRon1998}. 
We first explain the main ideas of the proof in \cite{GoldreichGoldwasserRon1998}, and then describe what modifications are needed for our case.
\paragraph{A sketch of the proof of Goldreich, Goldwasser and Ron \cite{GoldreichGoldwasserRon1998}}
\begin{itemize}
	\item The following observation is a key to the proof: Given a partition $P = (P_1, \ldots, P_k)$ of the set of vertices $V$ and a set $X$ which is small relatively to $V$, define the \emph{neighborhood profile} of a vertex $v \in X$ with respect to $P, X$ as the (ordered) set of $k$ densities of the edges from $v$ to each of the parts $P_j \setminus X$. The observation is that if all vertices of $X$ have approximately the same neighborhood profile, and if we redistribute the vertices of $X$ among the sets $P_1, \ldots, P_k$ so that each set receives roughly the same amount of vertices it lost to $X$, then the amount of edges between every pair of sets $P_i, P_j$ is roughly maintained.
	\item Generally we will deal with sets $X$ containing vertices with different neighborhood profiles, and will need a way to cluster them according to their profiles, and then be able to use the above observation. For this, one needs an \emph{oracle} that, given a vertex $v$, will determine efficiently and with good probability a good approximation of the neighborhood profile of $v$. Another related oracle that we need is one that efficiently approximates, for $P_1, \ldots, P_k$ and $X$, the ``$P_j$-fraction'' with respect to $X$, which determines what fraction of the vertices in $X$ with a given neighborhood profile belong to each $P_j \cap X$. 
	\item Using the oracles, it is shown that if a given graph satisfies a $\Phi$-instance, then the following process generates, with good probability, an explicit partition $P_1^s, \ldots, P_k^s$ that approximately satisfies $\Phi$. Assume for now that we start with a partition $P_1^0, \ldots, P_k^0$ that satisfies the $\Phi$-instance exactly.
	We partition all vertices of the graph into a large enough constant number of sets $X_1, \ldots, X_s$ of equal size. Now we do the following for $i=1, \ldots, s$: We take the elements of $X_i$, apply the oracles on them, accordingly approximate how many elements from $X_i$ with a certain neighborhood profile came from each $P_j^{i-1}$, and then ``shuffle'': Return the same amount of elements from $X_i$ with this profile to $P_j^{i-1}$, to create the part $P_j^i$ (the returned elements are chosen arbitrarily among those with the relevant profile, and in particular, are not necessarily the ones that were taken from $P_j^{i-1}$).
	\item There are two problems with the above statement. First, we do not know in advance the partition that satisfies the desired $\Phi$-instance, and thus along the way the partition $P_1, \ldots, P_k$ is not known to us. Second, we still do not know how to simulate the oracles. The solution to both of these problems is a brute force one: For each $X_i$ we pick a large enough constant size set $U_i \subseteq V \setminus X_i$, and then enumerate on all possible partitions of $U_i$ into $U_i \cap P_1^{i-1}, \ldots, U_i \cap P_k^{i-1}$ and all (rounded) possible values of the $P_j^{i-1}$-fraction for each $j=1,\ldots,k$ and all $i$.
	As it turns out, if there is a partition of $G$ satisfying the $\Phi$-instance, then our brute force search will find a good approximation of t with good probability.
	\item To turn the partitioning algorithm into a test, the observation is that one does not need to apply the first oracle on every vertex in each $X_i$ to determine its neighborhood profile. 
	Instead, we only apply it for a constant-size $S_i \subseteq X_i$ chosen at random. As it turns out, this process is almost as accurate as the partitioning process, and in particular, it is shown that if $G$ has a $\Phi$-instance then the process will accept, with good probability, a set of parameters of a $\Phi'$-instance which is close to the $\Phi$-instance. On the other hand, if $G$ is far from having such a $\Phi$-instance, then the process will reject, with good probability, all sets of parameters that are close to the $\Phi$-instance. This concludes the proof of \cite{GoldreichGoldwasserRon1998}.
\end{itemize}

\paragraph{Adapting the proof to our case}
\begin{itemize}
\item The first and minor issue that we have to deal with is the fact that our graphs are edge-colored, and not standard graphs as in \cite{GoldreichGoldwasserRon1998}. To handle this, instead of considering the neighborhood profile of a vertex, we are interested in the \emph{colored neighborhood profile} of a vertex $v$, which keeps, for any relevant part $P_j^i$ and any color $\sigma$, the fraction of vertices $u \in P_j^i$ for which $vu$ is colored $\sigma$. The rest of the proof translates naturally, implying that with this modification, the proof of \cite{GoldreichGoldwasserRon1998} also applies to edge-colored graphs.
\item The second issue is that our desired partition that satisfies the $\Phi$-instance has to be a refinement of the interval partition $I_1, \ldots, I_r$ of the input graph, as opposed to the situation in \cite{GoldreichGoldwasserRon1998}. This issue is also not hard to handle. A ``shuffle'' operation in the unordered case was the process of removing elements from $P_{j}^{i-1}$ into $X_i$, and then returning other elements from $X_i$ to create $P_j^i$. In our case we will have to make shuffles of elements separately within each $I_i$, since it is not allowed to move elements between different $I_i$'s. The rest of the analysis is essentially the same as in the proof of \cite{GoldreichGoldwasserRon1998}.
\item For the analysis of the last bullet to hold, we need the ability to pick a vertex uniformly at random from a given predetermined part $V_i$. This means that our algorithm is a piecewise canonical one, but not necessarily canonical. However, the transformation from a piecewise canonical test to a canonical one, that was proved in Section \ref{sec:piecewise_to_canonical}, implies the canonical testability of our version of the partition problem.
\end{itemize}

\section{Canonical testability versus regular reducibility}
\label{sec:canonical_vs_regular_reducibility}
As in the previous section, first we describe how the equivalence between testability and regular reducibility is proved in the unordered case \cite{AlonFischerNewmanShapira2009}, and then detail the small changes required to prove the edge-colored ordered case, namely Theorem \ref{thm:canonical_vs_regular_reducible}.

\subsection{The unordered case}
\subsubsection{Enhancing regularity efficiently}
\label{subsubsec:enhancing_regularity}
In Section 3 of \cite{AlonFischerNewmanShapira2009}, it is shown that if a pair of vertex sets $A,B$ has density close to $\eta$ and its regularity measure is very close to $\gamma$, then by making a small number of edge modifications (insertions/deletions), one can turn the pair $A,B$ into a ``perfect'' one, that has density exactly $\eta$ and is $\gamma$-regular. The proof has two main steps: In the first step, we take a ``convex combination'' of $G[A,B]$ with a random bipartite graph with density $\eta$. This process does not change significantly the density between $A$ and $B$, but since a random graph is highly regular, the combination is slightly more regular then the original $G[A,B]$, and this is all we need.
In the second step, we fix the density between $A$ and $B$ to be exactly $\eta$. This might very slightly hurt the regularity, but if in the first step we make $G[A,B]$ a bit more regular, i.e., $\gamma'$-regular for a suitable $\gamma' < \gamma$, then it will remain $\gamma$-regular even after the loss of regularity in the second step.

\subsubsection{Canonical testability implies regular reducibility}
\label{subsubsec:canonical_to_regularityr_reducible}
The easier direction of the proof is to show that any canonically testable property is also regular reducible, as is shown in Section 4 of \cite{AlonFischerNewmanShapira2009}.
Recall that, as discussed in Subsection \ref{subsubsec:signatures}, a regular enough partition of a graph $G$ provides a good approximation of the $q$-statistic of $G$. We consider a canonical test $T$ with a small enough proximity parameter, making $q$ vertex queries. 
Basically, our set of ``accepting'' regular instances (see Definition 2.6 in \cite{AlonFischerNewmanShapira2009}) will be created as follows: Initially, we take an $\epsilon$-net of possible parameters of regular partitions: This is a constant size quantized collection of the possible parameters of regular partitions, that ``represents'' all possible choices of parameters (in the sense that any possible choice of parameters has a representative in the constant size collection that is very close to it). Among the representatives from the $\epsilon$-net, we choose as accepting only those choices of parameters that predict acceptance of the above canonical test with probability at least $1/2$. 
Now, if a graph $G$ satisfies our property $\mathcal{P}$, then it is accepted with probability $2/3$ by the canonical test, and thus a regular enough partition of $G$ will be similar to some accepting regularity instance, making $G$ very close to satisfying this instance. Conversely, if $G$ is $\epsilon$-far from $\mathcal{P}$, then it must also be far from any graph $G'$ satisfying the accepting instance -- since, by our choice of the accepting regular instances as those that indicate acceptance of the canonical test, any such $G'$ is accepted by the canonical test with probability that is larger than $1/3$, meaning that $G'$ cannot be far from satisfying $\mathcal{P}$ (and thus $G$ cannot be too close to $G'$, otherwise it would be $\epsilon$-close to $\mathcal{P}$, a contradiction)

\subsubsection{Sampling preserves regular partitions}
\label{subsubsec:sampling_partitions}
In Section 5 of \cite{AlonFischerNewmanShapira2009} it is shown that if we sample a constant size set $S$ of vertices in a graph $G$, then with good probability the induced subgraph $G[S]$ will have $\gamma$-regular partitions with the same structure and approximately the same parameters (up to small differences) as those of $G$. The proof builds on a weaker argument of the same type, proved in \cite{Fischer2005}, which states that for a regular enough partition $P$ of $G$, and a large enough sample $S$, with good probability $S$ has a partition with roughly the same densities as these of $P$, and with regularity that is slightly worse than that of $P$. 

\subsubsection{Regular reducibility implies testability}
Due to the fact that canonical testability implies estimability, as we have seen in Section \ref{sec:canonical_to_estimable}, it is enough to show that satisfying a specific regularity instance is testable. To do so, we take a large enough sample $S$ of vertices and determine all possible parameters of regular partitions of $S$. By Subsection \ref{subsubsec:sampling_partitions}, these are essentially also all possible parameters of regular partitions of $G$, up to a small error. By Subsection \ref{subsubsec:enhancing_regularity}, this small error is not a problem, implying that we are able to determine (with good probability) whether $G$ satisfies the regularity instance by checking if it is close to one of the regular partitions suggested by $S$.

\subsection{Adapting the proof to the ordered case}
First, we need to translate the results from Section 3 in \cite{AlonFischerNewmanShapira2009} to the multicolored setting. The main lemma that we need here is the following.
\begin{lemma}[Ordered analogue of Lemma 3.1 in \cite{AlonFischerNewmanShapira2009}]
\label{lem:enhancing_regularity_multicolor}
There exists a function $f \colon \mathbb{N} \times (0,1) \to \mathbb{N}$  such that for any $0 < \delta \leq \gamma \leq 1$ and finite alphabet $\Sigma$ the following holds: Suppose that $(A,B)$ is a $(\gamma+\delta)$-regular pair of sets of vertices with density between $\eta - \delta$ and $\eta + \delta$ in a $\Sigma$-edge-colored graph, where $|A| = |B| = m \geq m_0(\eta, \delta, |\Sigma|)$. Then, it is possible to make at most $\delta f(|\Sigma|, \gamma) m^2$ edge color modifications in $G$, turning $(A,B)$ into a $\gamma$-regular pair with density precisely $\eta$.
\end{lemma}
The proof of Lemma \ref{lem:enhancing_regularity_multicolor} is largely similar to that of Lemma 3.1 in \cite{AlonFischerNewmanShapira2009}.
The only places that require special attention in the translation of the proof  are those with ``coin flip'' arguments, such as the one in the proof of Lemma 3.3 in \cite{AlonFischerNewmanShapira2009}. Adapting this type of arguments to the multicolored case is done as described in Subsection \ref{subsec:adapt_canonical_estimable}. In the proof of Lemma 3.3, for example, the second coin flip needs to have $|\Sigma|$ possible outcomes instead of two (where the probability to get a $\sigma$ should correspond to the desired density $\eta_\sigma$).

The corollary of Lemma \ref{lem:enhancing_regularity_multicolor} that is used in our proof is the following. Note that the notation in the following statement is largely borrowed from Definition \ref{def:ordered_regularity_instance}.
\begin{lemma}[Ordered analogue of Corollary 3.8 from \cite{AlonFischerNewmanShapira2009}]
There exists a function $\tau \colon \mathbb{N} \times (0,1) \to (0,1)$ for which the following holds.
Let $R$ be an ordered regularity instance as in Definition \ref{def:ordered_regularity_instance}, with the parameter $k$ in $R$ being large enough (as a function of the other parameters).
Suppose that for some $\epsilon > 0$, a $\Sigma$-edge-colored ordered graph $G$ has an equipartition $(V_{11}, \ldots, V_{1k}, \ldots, V_{r1}, \ldots, V_{rk})$ which is an $r$-refinement, and satisfies $|d_{\sigma}(V_{ij}, V_{i'j'}, \sigma) - \eta_{ij}^{i'j'}(\sigma)| \leq \epsilon \tau(|\Sigma|, \gamma)$ for all $i < i' \in [r]$, $j,j' \in [k]$, and $\sigma \in \Sigma$, and whenever $(i,j,i',j') \notin \bar R$, the pair $V_{ij}, V_{i'j'}$ is $(\gamma + \epsilon \tau(|\Sigma|, \gamma))$-regular.
Then $G$ is $\epsilon$-close to satisfying $R$.
\end{lemma}

\subsubsection*{Canonical testability to ordered regular reducibility}
The next step is to show that any canonically testable ordered graph property is (ordered) regular reducible. 
Recall that, by Section \ref{sec:canonical_to_estimable}, canonical testability implies $(\epsilon, q(\epsilon), k)$-canonicality for $k$ large enough (with respect to $q(\epsilon)$), so it is enough to show the following.
\begin{lemma}[Ordered analogue of Lemma 4.1 in \cite{AlonFischerNewmanShapira2009}]
If a property $\mathcal{P}$ is $(\epsilon, q(\epsilon), k)$-canonical for any $\epsilon$ and any $k$ large enough with respect to $q(\epsilon)$, then it is ordered regular reducible.
\end{lemma}
For the results of Section 4 in \cite{AlonFischerNewmanShapira2009}, we define the ordered multicolored analogues of Definitions 4.3 and 4.7 in \cite{AlonFischerNewmanShapira2009} as follows. Note the following ``notational glitch'': $\sigma$ in our definition refers to an edge color, whereas in Definition 4.3 of \cite{AlonFischerNewmanShapira2009} it plays a totally different role, as a permutation.
\begin{definition}
Let $H = (U, E_H)$ be a $\Sigma$-edge-colored ordered graph on $h$ vertices $u_1 < \ldots < u_h$, and let $W = (U, E_w)$ be an (edge) weighted $\Sigma$-edge-colored ordered graph on $h$-vertices, where the weight of edge $(u_i,u_j)$ is $\eta_{ij}$. Define
$IC(H, W) = \prod_{\sigma \in \Sigma} \prod_{u_i u_j \in E_H^{-1}(\sigma)} \eta_{ij}$.

For an ordered regularity instance $R$ as in Definition \ref{def:ordered_regularity_instance}, define $IC(H,R) = \sum_{W \in \mathcal{W}} IC(H,W)$ where $\mathcal{W}$ ranges over all $q$-vertex weighted $\Sigma$-edge-colored weighted graph of the following type: Pick $q$ pairs $(i_1, j_1), \ldots, (i_q, j_k)$ with $i_1 < \ldots < i_q \in [r]$ and $j_1, \ldots, j_q \in [k]$, and take $W$ to be the graph in which the weight of color $\sigma$ between vertices $u_a < u_b$ is $\eta_{i_a, j_a}^{i_b, j_b}(\sigma)$.
\end{definition}
With these definitions, it is straightforward to translate the results of Section 4 in \cite{AlonFischerNewmanShapira2009} to our setting. Note that an analogue for Definition 4.5 in that section is not needed in our case, since there are no non-trivial automorphisms in an ordered graph.
In the proof of Lemma 4.1 in \cite{AlonFischerNewmanShapira2009}, let $\mathcal{A}$ be the family of edge-colored ordered graphs on $q = q(\epsilon)$ vertices, promised to us through Definition \ref{def:k_canonical} by the fact that our given property $\mathcal{P}$ is $(\epsilon, q, k)$-canonical, for $k$ that is sufficiently large. 
As in the unordered case, we take a (constant size) set $\mathcal{I}$ of ordered regular instances, such that any possible regular instance has parameters that are very close to one of the instance in $\mathcal{I}$. Our chosen $\mathcal{R}$ in Definition \ref{def:ordered_regular_reducible} will be as in the unordered case: $\mathcal{R} = \{R \in \mathcal{I} : \sum_{H \in \mathcal{A}} IC(R,H) \geq 1/2 \}$. The rest of the proof goes as in the unordered case.

\subsubsection*{Ordered regular reducibility to (piecewise) canonical testability}
It follows from the definition of regular reducibility, similarly to the unordered case, that it is enough to show that the property $\mathcal{P}$ of satisfying a given regularity instance is canonically testable (the easy proof of the analogous unordered statement appears in Section 6 of \cite{AlonFischerNewmanShapira2009}, and translates directly to our case). In fact, by Lemma \ref{lem:tocanonical}, it is enough to show that $\mathcal{P}$ is piecewise canonically testable. 
Indeed, the core of the proof of this statement in the unordered case is in the fact that for $\gamma$, a large enough (as a function of $\gamma$) sample of a graph has, with good probability, essentially the same $\gamma$-regular equipartitions as the containing graph, up to a small error.

The definition of \emph{similar} regular partitions in the ordered case (analogous to Definition 5.1 from \cite{AlonFischerNewmanShapira2009}) is the same as in the unordered case, but it refers to $(\gamma, k)$-regular partitions, instead of the unordered $\gamma$-regular ones. The analogue of Lemma 5.2 in the ordered case is exactly the same, except that we require the sample $Q$ to have exactly $q$ vertices in each interval of the $k$-interval equipartition (note that this is doable using piecewise-canonical algorithms). The proofs from this section (including the proof of the weaker result from \cite{Fischer2005}), as well as the proof of Theorem 1 from Section 6, translate readily to the ordered case.


\appendix
\section{Sparse boundary implies earthmover resilience}
\label{app:sparse_boundary}
Here we present the proof of Theorem \ref{thm:sparse_boundary_to_earthmover_resilience}.
Let $\mathcal{I}$ be an $n \times n$ black-white image with a $c$-sparse boundary (recall Definition \ref{def:sparse_boundary}). Without loss of generality, we may assume that all pixels in locations $\{1,n\} \times [n] \cup [n] \times \{1,n\}$ are white; otherwise, we may replace $c$ with $c+4$ and turn $\mathcal{I}$ into an $(n+2) \times (n+2)$ image by adding an ``artificial'' white boundary to the image.  
To prove Theorem \ref{thm:sparse_boundary_to_earthmover_resilience}, it is enough to show that for any image $\mathcal{I}'$ that is the result of making at most $\delta n^2$ basic moves on $\mathcal{I}$, the absolute Hamming distance between $\mathcal{I}$ and $\mathcal{I}'$ is $O(c \sqrt{\delta} n^2)$. 

The following lemma suggests that it is enough to prove a similar statement for an image $\mathcal{J}$ of our choice that is close enough (in Hamming distance)
to $\mathcal{I}$.
\begin{lemma}
\label{lem:sparse_Hamming_change}
Fix $\alpha, \beta > 0$ and let $\mathcal{I}, \mathcal{J} \colon [n] \times [n] \to \Sigma$. Suppose that $d_H(\mathcal{J}, \mathcal{J}') \leq \alpha$ for any $\mathcal{J}' \colon [n] \times [n] \to \Sigma$ that satisfies $d_e(\mathcal{J}, \mathcal{J}') \leq \beta$. Then $d_H(\mathcal{I}, \mathcal{I}') \leq \alpha + 2 d_H(\mathcal{I}, \mathcal{J})$ for any $\mathcal{I}' \colon [n] \times [n] \to \Sigma$ satisfying $d_e(\mathcal{I}, \mathcal{I}') \leq \beta$. 
\end{lemma}
\begin{proof}
Write $\gamma = d_H(\mathcal{I}, \mathcal{J})$. 
Consider any $\mathcal{I}'$ satisfying $d_e(\mathcal{I}, \mathcal{I}') \leq \beta$ and let $\sigma$ be a minimal unordered isomorphism of $n \times n$ images\footnote{The formal definition is given for ordered graphs in Definition \ref{def:unordered_iso}, but can translated naturally to images using our standard representation of an image as an ordered graph.} that maps $\mathcal{I}$ to $\mathcal{I}'$. By the minimality of $\sigma$, the image $\mathcal{J}' = \sigma(\mathcal{J})$ satisfies $d_e(\mathcal{J}, \mathcal{J}') \leq \beta$ and so $d_H(\mathcal{J}, \mathcal{J}') \leq \alpha$. 
On the other hand, we know that $d_H(\mathcal{I}', \mathcal{J}') = d_H(\sigma(\mathcal{I}), \sigma(\mathcal{J})) = d_H(\mathcal{I}, \mathcal{J}) = \gamma$ where the least equality follows from the fact that Hamming distance between two images is preserved when applying the same unordered isomorphism on both of them. 
The triangle inequality for the Hamming distance implies that
$$
d_H(\mathcal{I}, \mathcal{I}') \leq d_H(\mathcal{I}, \mathcal{J}) + d_H(\mathcal{J}, \mathcal{J}') + d_H(\mathcal{J}', \mathcal{I}') \leq \beta + 2 \gamma
$$
as desired.
\end{proof}
Indeed, Lemma \ref{lem:sparse_Hamming_change} implies that in order to prove Theorem \ref{thm:sparse_boundary_to_earthmover_resilience}, it is enough to show that there exists some $n \times n$ black-white image $\mathcal{J}$ with $d_H(\mathcal{I}, \mathcal{J}) = O(c \sqrt{\delta} n^2)$, such that for any image $\mathcal{J}'$ that is the result of making at most $\delta n^2$ basic moves on $\mathcal{J}$, we have $d_H(\mathcal{J}, \mathcal{J}') = O(c \sqrt{\delta} n^2)$.
In order to explain which $\mathcal{J}$ to take (as a function of $\mathcal{I}$), and proceed with the rest of the proof, we need several topological definitions.
A \emph{pixel} $P = (i,j)$ in $\mathcal{I}$ is represented by its location $(i,j)$, and its color (black/white) is denoted $\mathcal{I}[P]$. 
The \emph{distance} between two pixels $(i,j), (i',j') \in [n] \times [n]$ is defined as $|(i,j) - (i',j')| = |i-i'| + |j-j'|$; these pixels are \emph{neighbors} if the distance between them is $1$. 
A \emph{shape} $\mathcal{S}$ in $\mathcal{I}$ is a connected component (with respect to the neighborhood relation) of pixels with the same color. We call $P^0 = (1,1)$ the \emph{outer pixel} of an image, and the shape $S^0$ that contains it is called the \emph{outer shape}. Note that, by our assumption, the outer shape of $\mathcal{S}$ contains all pixels in $(\{1,n\} \times [n]) \cup ([n] \times \{1,n\})$.

A \emph{path} between pixels $P$ and $P'$ is a tuple of (not necessarily disjoint) pixels $P_1 = P, P_2, \ldots, P_t = P'$ in $\mathcal{I}$, such that $P_s$ and $P_{s+1}$ are neighbors for any $1 \leq s \leq t-1$.
The \emph{outer boundary} $B(S)$ of a shape $S \neq S^0$ is the set of all pixels $P$ in $S$ satisfying the following: there exists a path from $P^0 = (1,1)$ to $P$ that does not intersect $S \setminus \{P\}$. Finally, a pixel $P$ is \emph{encircled} by a shape $S$ if any path from $(1,1)$ to $P$ intersects $S$ (this includes all pixels $P \in S$). If all pixels $P$ encircled by $S$ satisfy $P \in S$, we say that $S$ is \emph{full}. 

Our first lemma states that if two neighboring pixels have different colors, than one of them lies in the outer boundary of its shape.
\begin{lemma}
\label{lem:boundary_pixel}
Let $P_1, P_2$ be two neighboring pixels, where $P_1$ is black and lies in shape $S_1$ and $P_2$ is white and lies in $S_2$. Then either $P_1 \in B(S_1)$ or $P_2 \in B(S_2)$ (or both).
\end{lemma}
\begin{proof}
If there exists a path from $(1,1)$ to a pixel $P'_1$ in $S_1$, that does not intersect $S_2$, then $P_2 \in B(S_2)$. To see this, recall that $S_1$ is connected (by definition of a shape) and thus there exists a path from $P'_1$ to $P_1$ that remains inside $S_1$. Concatenating the above two paths and adding $P_2$ at the end implies that $P_2 \in B(S_2)$.

Otherwise, all paths from $(1,1)$ to any pixel in $S_1$ intersect $S_2$. In particular, this implies that there exists a path from $(1,1)$ to some $P'_2 \in S_2$ that does not intersect $S_1$. Symmetrically to the previous paragraph, we get that $P_1 \in B(S_1)$.
\end{proof}
define $B(\mathcal{I})$ as the union of all outer boundaries $B(S)$ where $S$ ranges over all shapes in $\mathcal{I}$ other than $S^0$.
The next lemma follows immediately from Lemma \ref{lem:boundary_pixel} and the fact that $\mathcal{I}$ is $c$-sparse.
\begin{lemma}
\label{lem:sum_of_outer_boundaries}
$|B(\mathcal{I})| \leq 4cn$, where $S$ ranges over all shapes in $\mathcal{I}$ other than $S^0$.
\end{lemma}

The next lemma implies that shapes with a small boundary cannot encircle a large number of pixels. This will play a crucial role in the design of $\mathcal{J}$.
\begin{lemma}
	\label{lem:sparse_boundary_small_shapes}
	The total number of pixels encircled by a shape $S \neq S^0$ is at most $|B(S)|^2$.
\end{lemma}
\begin{proof}
	We may assume that $S$ is full.
	Let $r(S)$ denote the number of pairs of neighboring pixels $(P, P')$ where $P \in S$ and $P' \notin S$. Then $r(S) \leq 4 |B(S)|$. Among all possible full shapes $S$ with a given value of $r(S)$, an (axis-aligned) rectangle contains the biggest number of pixels. This follows by iterating the following simple type of arguments as long as possible: If $(i, j)$ and $(i+1, j+1)$ are pixels of $S$ while $(i, j+1) \notin S$, then adding $(i, j+1)$ to $S$ yields a shape $S'$ with more pixels than in $S$, that satisfies $r(S') \leq r(S)$.
	
	Now note that the number of pixels in a rectangle $S$ is bounded by $r(S)^2 / 16 \leq |B(S)|^2$. The bound is achieved if $S$ is a square with side length $r(S)/4$. 
\end{proof}

We pick $\mathcal{J}$ using the following iterative process. Start with $\mathcal{J} = \mathcal{I}$, and as long as possible do the following: Take a shape $S \neq S^0$ in $\mathcal{J}$ with $|B(S)| \leq \sqrt{\delta} n$, and recolor all pixels encircled by $S$ by the opposite color to that of $S$; repeat. Each such iteration deletes all pixels of $B(S)$ from $B(\mathcal{J})$ (and does not add any new pixels to $B(\mathcal{J})$), modifying at most $|B(S)|^2$ pixels in $\mathcal{J}$, so by Lemmas \ref{lem:sum_of_outer_boundaries} and \ref{lem:sparse_boundary_small_shapes}, in the end of the process we have $d_H(\mathcal{I}, \mathcal{J}) = (4cn / \sqrt{\delta} n) \cdot O(\delta n^2) = O(c \sqrt{\delta} n^2)$ as desired. 

Consider any composition $\sigma$ of at most $\delta n^2$ basic moves on $\mathcal{J}$. The new location of any pixel $P$ after the basic moves is denoted by $\sigma(P)$. 
To conclude the proof, we need to show that the number of pixels $P$ for which $\mathcal{J}[P] \neq \mathcal{J}[\sigma(P)]$ is $O(c \sqrt{\delta} n^2)$.

Define the \emph{boundary distance} of a pixel $P$ in $\mathcal{J}$ as the minimal distance of $P$ to a pixel from $B(\mathcal{J})$.
Our next lemma states that $\sigma$ can only change the color of a small number of pixels with large boundary distance.  
\begin{lemma}
	\label{lem:far_from_sparse_boundary}
	No more than $O(\sqrt{\delta} n^2)$ pixels $P$ in $\mathcal{J}$ have boundary distance at least $\sqrt{\delta} n$ and satisfy $\mathcal{J}[P] \neq \mathcal{J}[\sigma(P)]$.
\end{lemma}
\begin{proof}
By Lemma \ref{lem:boundary_pixel}, a pixel $P$ with boundary distance $d$ that satisfies $J[P] \neq J[\sigma(P)]$ must either be contained in a row that was moved at least $d/2$ times or a column that was moved at least $d/2$ times by the basic moves of $\sigma$; here we pick $d = \sqrt{\delta} n$. With $\delta n^2$ basic moves, at most $O(\sqrt{\delta}n)$ rows and columns can be moved $\sqrt{\delta} n / 2$ or more steps away from their original location. The total number of pixels in these rows and columns is $O(\sqrt{\delta} n^2)$, as desired. 
\end{proof}

It remains to show that no more than $O(c \sqrt{\delta} n^2)$ pixels in $\mathcal{J}$ have boundary distance less than $\sqrt{\delta} n$.
The following lemma serves as a first step towards this goal.
\begin{lemma}
\label{lem:cover_boundary_by_path}
Let $S \neq S^0$ be a shape in $\mathcal{J}$. Then there exists a path $\Gamma(S)$ (possibly with repetitions of pixels) of length $O(|B(S)|)$, that covers all pixels of $B(S)$.
\end{lemma}
\begin{proof}
Consider an $n \times n$ grid in $\mathbb{R}^2$ where the pixel $(i,j)$ is represented by the unit square whose four endpoints are $\{i-1, i\} \times \{j-1, j\}$. Since any shape $S$ is connected (by definition) under the neighborhood relation, in this representation $S$ is the interior of a closed curve consisting of at most $4|B(S)|$ axis-parallel length-1 segments.
Following the segments of this curve in a clockwise fashion and recording all pixels in $S$ that we see on our right (including pixels that we only visit their corner) constructs a path (possibly with repetitions) that contains only the pixels of $B(S)$ and some of their neighbors; recall that each pixel in $\mathcal{J}$ has at most four neighbors. Moreover, each pixel appears at most $O(1)$ times in this path, and so the total length of the path is $O(|B(S)|)$.
\end{proof}

Finally, the next lemma allows us to conclude the proof.
\begin{lemma}
\label{lem:sparse_boundary_small_distance}
Let $S \neq S^0$. The number of pixels in $\mathcal{J}$ of distance at most $d$ to $B(S)$ is  $O(d|B(S)|+d^2)$. 
\end{lemma}
\begin{proof}
Take the path $\Gamma(S)$ obtained in Lemma \ref{lem:cover_boundary_by_path}. For each pixel $P \in \Gamma(S)$ let $B_d(P) = \{P' \in [n] \times [n]: |P'-P| \leq d\}$ denote the \emph{$d$-ball} around $P$ in $\mathcal{I}$. Note that the set of all pixels of distance at most $d$ to $B(S)$ is contained in $\cup_{P \in \Gamma(S)} B_d(P)$. 
Trivially, $B_d(P)$ contains at most $d^2$ pixels for any $P$. Moreover, if $P_1$ and $P_2$ are neighbors, then $|B_d(P_1) \setminus B_d(P_2)| \leq d$. The statement now follows since $\Gamma(S)$, a path, is connected under the neighborhood relation, and is of length $O(|B(S)|)$.
\end{proof}
Recall that $|B(\mathcal{J})| \leq |B(\mathcal{I})| \leq 4cn$ by Lemma \ref{lem:sum_of_outer_boundaries}. Since all shapes $S \neq S^0$ in $\mathcal{J}$ satisfy $|B(S)| > \sqrt{\delta} n$, the number of such shapes must be at most $4c / \sqrt{\delta}$. Lemma \ref{lem:sparse_boundary_small_distance} implies that the total number of pixels of boundary distance at most $d = \sqrt{\delta} n$ in $\mathcal{J}$ is at most $O(dcn + d^2 c / \sqrt{\delta}) = O(c \sqrt{\delta} n^2)$. Along with Lemma \ref{lem:far_from_sparse_boundary}, this completes the proof of Theorem \ref{thm:sparse_boundary_to_earthmover_resilience}.

\end{document}